\documentclass[journal]{IEEEtran}
\usepackage{amsthm}%proof
\usepackage{mathtools}
\usepackage{epsfig}
\usepackage{mathrsfs}
\usepackage[ruled]{algorithm}
\usepackage{algorithmic}
\usepackage{amsfonts}
\usepackage{graphicx}
\usepackage{url}
\usepackage{bm}
\usepackage{color}
\usepackage{xcolor}
\usepackage{subfig}
\newtheorem{theorem}{Theorem}
\newtheorem{definition}{Definition}
\newtheorem{lemma}{Lemma}
\usepackage{cite}

\begin{document}

\title{Crowdsensing Game with Demand Uncertainties: A Deep Reinforcement Learning Approach}

\author{Yufeng~Zhan,
        Yuanqing Xia,~\IEEEmembership{Senior~Member,~IEEE,}
        Jiang Zhang,
        Ting Li, and
        Yu~Wang,~\IEEEmembership{Fellow,~IEEE}
\thanks{Yufeng Zhan, Yuanqing Xia and Jiang Zhang are with the School of Automation, Key Laboratory of Intelligent Control and Decision of Complex
Systems, Beijing Institute of Technology, Beijing 100081, P. R. China. E-mail: zhanyf1989@gmail.com (Zhan), xia\_yuanqing@bit.edu.cn (Xia),~bitzj2015@outlook.com (Zhang).}
\thanks{Ting Li and Yu Wang are with the Wireless Networking and Sensing (WiNS) Lab, Department of Computer Science, University of North Carolina at Charlotte, Charlotte, NC 28223, USA. E-mail: tli8@uncc.edu (Li), yu.wang@uncc.edu (Wang).}
}

\maketitle

\begin{abstract}
Currently, explosive increase of smartphones with powerful built-in sensors such as GPS, accelerometers, gyroscopes and cameras has made the design of crowdsensing applications possible, which create a new interface between human beings and life environment. Until now, various mobile crowdsensing applications have been designed, where the crowdsourcers can employ mobile users (MUs) to complete the required sensing tasks. In this paper, emerging learning-based techniques are leveraged to address crowdsensing game with demand uncertainties and private information protection of MUs. Firstly, a novel economic model for mobile crowdsensing is designed, which takes MUs' resources constraints and demand uncertainties into consideration. Secondly, an incentive mechanism based on Stackelberg game is provided, where the sensing-platform (SP) is the leader and the MUs are the followers. Then, the existence and uniqueness of the Stackelberg Equilibrium (\textbf{SE}) is proven and the procedure for computing the \textbf{SE} is given.
Furthermore,
%to protect MUs' private information,
a dynamic incentive mechanism (DIM) based on deep reinforcement learning (DRL) approach is investigated without knowing the private information of the MUs. It enables the SP to learn the optimal pricing strategy directly from game experience without any prior knowledge about MUs' information. Finally, numerical simulations are implemented to evaluate the performance and theoretical properties of the proposed mechanism and approach.
\end{abstract}


\IEEEpeerreviewmaketitle

\section{Introduction}
\IEEEPARstart{W}{ith} the ubiquity of mobile devices such as smartphones and tablets that are equipped with multiple powerful built-in sensors including GPS, accelerometer, gyroscope, camera, etc., the mobile crowdsensing (MCS) applications which provide location based services \cite{guo2015mobile} become possible. Currently, various of MCS systems \cite{mohan2008nericell,thiagarajan2009vtrack,cheng2014aircloud} have been deployed that cover almost every aspect of our lives, including healthcare, intelligent transportation, environmental monitoring, etc.

In the MCS system that offers crowdsensing applications, the sensing-platform (SP) will recruit mobile users (MUs) at locations of interest to report sensing data. Many of existing MCS systems \cite{xiao2017online,rana2010ear} are based on the voluntary participation from MUs. However, to perform the sensing tasks, the participating MUs have to consume their own resources such as computing and communicating energy. Moreover, the MUs may face the potential privacy threats when the sensing data is submitted with own sensitive information (e.g. location tags and visiting patterns). For these reasons, the MUs would not be interested in participating in the sensing tasks unless they receive a satisfying reward to compensate their resources consumption and potential privacy breach. Therefore, it is necessary to design an effective incentive mechanism that can stimulate the MUs to participate in the crowdsensing applications. In order to achieve the maximum user participation level, large quantities of incentive-aware mechanisms \cite{zhan2018incentive,duan2017distributed,he2017exchange,yang2016incentive,zhang2014free} have been proposed by research community for the MCS systems. Notably, in real practice, the smart devices' resources such as energy are limited, and these resources need to satisfy MUs' varying demand caused by their uncertain behavior (e.g., when MUs are busy at work, their smart devices may be free. When MUs want to have entertainments, their smart devices may be occupied with few resources left). However, few of these aforementioned works take MUs¡¯ uncertain behavior into consideration. Therefore, the design of incentive mechanism for MCS game with demand uncertainties is still an open problem.
%However, most of these aforementioned past works do not consider the MUs own resource demand uncertainties. This makes the study of crowdsensing game with demand uncertainties is still an open problem.
%
%In real practice, the smart devices' resources such as energy are limited, and these resources need satisfying the MUs own demand. Due to the uncertainty of MUs' behavior (e.g., when MUs are working busy, they have little time to operate their smart devices. Otherwise, they will have time to use the smart devices for entertainment which will consume more resource and have few resources left for serving the MCS systems), hence, when design the incentive mechanism for the MCS systems much take the MUs' demand uncertainties into consideration. In this paper, we formulate the interactions between SP and MUs as a two-stage Stackelberg game.

To deal with this problem, in this paper, the interaction between SP and MUs is formulated into a two-stage Stackelberg game. As shown in Fig.~\ref{fig_1}, in Stage I, the SP as the leader of the Stackelberg game first determines and broadcasts its pricing policy. In Stage II, each MU as a follower computes his or her sensing effort based on the price offered by the SP, his or her resources constraints and demand uncertainties. The analysis in this two-stage problem is particularly challenging, as we need to characterize the SP's profit by first computing the MUs' sensing effort with demand uncertainties. Through mathematical analysis, the existence and uniqueness of the Stackelberg Equilibrium (\textbf{SE}) in this MCS game is proven and the expressions for computing the \textbf{SE} is derived. That is, the SP in Stage I has an optimal pricing strategy and the MUs in Stage II also have optimal decisions under their own demand uncertainties.

\begin{figure*}[!t]
  \centering
  \includegraphics[width=16cm]{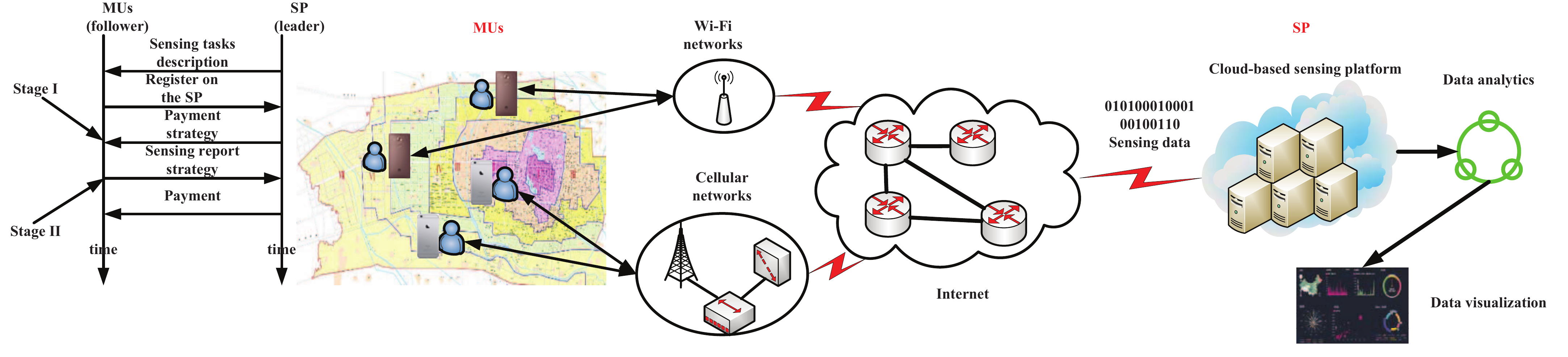}\\
  \caption{Illustration of MCS system.}\label{fig_1}
\end{figure*}

However, in order to compute the \textbf{SE} of the above static MCS game, the SP needs to know the private information of the MUs, which is impossible in lots of practical situations. To protect MUs' private information, the dynamic MCS game is modeled and dynamic incentive mechanism based on deep reinforcement learning (DRL) approaches are employed, which enable the sensing platform to learn the optimal pricing strategy directly from game experience (the past game records). Since the game experience of the SP can be regarded as a motivation for its future pricing strategy, the dynamic MCS game can be formulated into a Markov Decision Process (MDP) problem. Thus, it can be addressed by DRL algorithms effectively \cite{schulman2015trust}.
%Nevertheless, we give the analysis to compute the optimal MUs' decision under their demand uncertainties in Stage II, and give proof that there is a unique strategy for the SP in Stage I. However, in the static crowdsensing game, in order to compute the \textbf{SE} of the crowdsensing game, the SP needs to know the private information of the MUs, while this is impossible in lots of situations. \textbf{(Why we can formulate the crowdsensing game as a finite MDP? And how to compute the strategy in MDP? refer to \cite{xiao2017secure})}

Overall, the main contributions of this paper can be summarized as follows:
\begin{enumerate}
  \item A novel economic model for the MCS game with MUs' resources constraints and demand uncertainties is formulated and an incentive mechanism based on a two-stage Stackelberg game is designed.
  \item The existence and uniqueness of the \textbf{SE} in the proposed MCS game is proven and its computing procedure is provided, revealing the feasibility of allowing MCS game to cope with MUs' uncertain demand and limited resources.
  \item A dynamic incentive mechanism (DIM) based on DRL approach for the dynamic MCS game is proposed, which enables the SP to learn the optimal pricing strategy directly from game experience without any prior knowledge about MUs' private information.
      %Therefore, privacy protection of MUs in this MCS game is achieved.
  \item Numerical simulation results demonstrate the effectiveness of the proposed incentive mechanisms for both of the static MCS game and the dynamic MCS game. It is also derived that the demand uncertainties have a significant impact on MCS system performance.
  %Extensive simulation result shows that the incentive mechanism designed in this paper could obtain the optimal strategy, and the DRL-based approach could learn the optimal strategy with finite steps. And we can also derive that the demand uncertainties have a significant impact on MCS system performance.
\end{enumerate}

The rest of the paper is organized as follows. Section~\ref{sec_2} provides a literature review. Section~\ref{sec_3} presents the network economics model of the crowdsensing system. The incentive mechanism based on a two-stage Stackelberg game for the static MCS game is designed in Section~\ref{sec_4} and the DRL-based dynamic incentive mechanism for the dynamic MCS game is designed in Section~\ref{sec_5}. In Section~\ref{sec_6}, the numerical simulations are conducted to evaluate the performance of the proposed incentive mechanisms, followed by conclusions of this paper in Section~\ref{sec_7}.
%In Section~\ref{sec_4}, we present the a
%nalysis process of the static crowdsensing game. Then, based on the analysis results of the static crowdsensing game, in Section~\ref{sec_5}, we present the dynamic crowdsensing game based on DRL approach. Extensive numerical simulations are conducted to evaluate the performance of the proposed incentive mechanism in Section~\ref{sec_6}. At last, we conclude the paper in Section~\ref{sec_7}.

\section{Literature Review}\label{sec_2}
MCS has been widely studied in recent years \cite{guo2015mobile}. For example, Reddy et al. \cite{reddy2010examining} developed an application to enable sensing platform employ well-suited participants to complete sensing tasks. Xiao et al. \cite{xiao2017online} and Li et al. \cite{li2015dynamic} both studied the task allocation and participants selection problem in MCS. However, these works only focus on the user selection, task assignment or sensing data collection. They do not consider the design of incentive mechanism, which has been widely studied in lots of other fields such as spectrum trading \cite{gao2011spectrum,gao2011map}, routing \cite{ning2013self,wang2006ours}, cooperative communication \cite{chen2011conflicts,yang2012truthful}, and network security \cite{yang2013coping,xiao2015user}.

A user-centric MCS based on auction-based approach is proposed in \cite{yang2016incentive}, motivating MUs to participate in the MCS tasks. Lee and Hoh \cite{lee2010sell} proposed a reverse auction incentive mechanism based on dynamic prices in offline settings, where MUs can claim their bidding prices for the sensing data. Wen et al. \cite{wen2015quality} proposed an incentive mechanism which took the MUs' sensing quality into account, where the MUs are paid based on the quality of the sensing data instead of sensing time. Luo et al. \cite{luo2014profit} assumed that the cost distribution was known, then they designed an all-pay auction based incentive mechanism which can maximize the expected profit and meanwhile satisfied the individual rationality. Zhao et al. \cite{zhao2014crowdsource} proposed the online auction with budget constraints which applied the greedy task allocation strategy to achieve high energy efficiency with good fairness among MUs who arrived sequentially and randomly. Xu et al. \cite{xu2015incentive} proposed the incentive mechanisms for time window dependent tasks in mobile crowdsensing based on reverse auction and formulated the problem as the social optimization user selection problem.

There are also a few of studies on incentive mechanism design for MCS based on Stackelberg game. Yang et al. \cite{yang2016incentive} modeled the platform-centric incentive mechanism as a Stackelberg game. In their model, the SP has one task in a sensing slot and announces a total reward. The MUs decide their sensing strategy according to the total reward and other MUs' habits. Duan et al. \cite{duan2012incentive} used the Stackelberg game to design a threshold revenue model for the MUs. They studied two applications, data acquisition and distributed computing. For data acquisition, they took a threshold revenue model, in which a certain number of MUs are required to build the corresponding data base successfully. And for distributed computing and heterogeneous users, a contract-based mechanism had been designed to decide different task-reward combinations. Cheung et al. \cite{cheung2018delay} designed the delay-sensitive mobile crowdsensing based on Stackelberg game. In \cite{maharjan2016optimal}, Maharjan et al. proposed the multimedia application of crowdsensing based on Stackelberg game. However, these works did not take MUs' demand uncertainties into consideration, which is still a challenging problem since the MUs' resources are limited and their behaviors are uncertain in reality. A work related to this problem is \cite{zhan2018incentive}, where Zhan et al. studied the incentive mechanism design with demand uncertainties. However, \cite{zhan2018incentive} was based on one-to-many bargaining approach where the SP needed to cooperate with all the MUs. While in a free market, the SP does not know
the MUs and the MUs also do not know each other. Therefore, it is impossible to reach a partnership in a free market MCS.
%However, these MCS systems based on Stackelberg game approach all needed the private information of the participants. Xiao et al. \cite{xiao2015secure,xiao2017secure} designed the secure MCS, they still used the Stackelberg game to formulate the interaction between SP and MUs. In these two papers, they designed the Q-Learning and deep Q-Learning approaches to learn the optimal strategy of the SP and MUs, which the SP did not need know the private information of the MUs.

In addition, most of the existing MCS systems based on Stackelberg game approach require the private information of the participants. To deal with this problem, Xiao et al. \cite{xiao2017secure} designed the secure MCS, they still used the Stackelberg game to formulate the interaction between the SP and MUs. They designed the Q-Learning and deep Q-Learning approaches to learn the optimal strategy of the SP and MUs, in which the SP did not need know the private information of the MUs. It is noteworthy that their approaches can only cope with MCS game with discrete pricing strategy. When the pricing strategy is continuous, their methods are unable to work.
%Recently, both secure MCS game \cite{privacy-1,privacy-2} and learning-based participant selection \cite{learning-1} have also been studied separately.

Therefore, how to design the incentive mechanism for MCS game with MUs' demand uncertainties based on Stackelberg game in a free market and how to address private information protection problem of MUs under continuous pricing conditions become the focuses of this paper.

%However, because the MUs' resources are limited and their behaviors are uncertain. It is still challenging for the MCS system to stimulate the MUs with their own demand uncertainties. A particularly related paper to our work is \cite{zhan2018incentive}, where Zhan et al. studied the incentive mechanism design with demand uncertainties. However, \cite{zhan2018incentive} was based on one-to-many bargaining approach, in it the SP needed to cooperate with all the MUs, while in a free market, the SP did not know the MUs, also the MUs did not know each other. Therefore, it is impossible to reach a partnership in a free market crowdsensing. We have proposed the Stackelberg game approach to formulate the interactions between SP and MUs with demand uncertainties. To keep the privacy of the MUs, we also design the DRL approach to learn the optimal strategy for the SP and MUs which makes the SP does not need to know the private information of the MUs.

\section{System Model}\label{sec_3}
We consider a single SP which resides in the cloud and consists of some servers. A set $\mathcal{N}=\{1,2,\cdots,N\}$ of MUs that connect to the sensing-platform via the Internet. The sensing-platform will stimulate the mobile users to participate in the MCS tasks via rewards. More specifically, the SP's economic model is described in Section~\ref{sec_3-1} and the model of the MUs is presented in Section~\ref{sec_3-2}, followed by the problem formulation in Section~\ref{sec_3-3} finally.
%In this section, we first describe the SP's economic model in Section~\ref{sec_3-1}, and then present the model of the MUs in Section~\ref{sec_3-2} and problem formulation in Section~\ref{sec_3-3} finally.

\subsection{SP's Payoff}\label{sec_3-1}
The economic model for SP concentrates on the direct utility for the SP. That is, the utility increase only from the MUs' sensing level\footnote{Some indirect utilities include network effects which monotonically increase with the sensing level \cite{chen2016incentivizing}.}. Let $\phi(\cdot)$ denote the SP's utility. We will employ a generic utility function which is continuous, differentiable, strictly increasing and strictly concave. $\phi(\cdot)$ consists of two main characteristics of widely used utility functions \cite{duan2017distributed,he2017exchange}: i) $\phi(\cdot)$ increases with MUs' sensing level and ii) the growth rate of $\phi(\cdot)$ decreases with the MUs' sensing level increasing, i.e., $\frac{{\partial \phi (\cdot)}}{{\partial {x_i}}} \ge 0$ and $\frac{{{\partial ^2}\phi (\cdot)}}{{{\partial ^2}{x_i}}} < 0$.

Let $x_n$ denote MU $n$'s sensing resources contribution to SP, and $p_n \ge 0$ denote the SP's price to MU $n$. The sensing contribution profile and price profile are, respectively,
$$\bm{x} \buildrel \Delta \over = [x_1,x_2,\cdots,x_N]^T, \bm{p} \buildrel \Delta \over = [p_1,p_2,\cdots,p_N]^T.$$

Given $\bm{x}$ and $\bm{p}$, the total cost of the SP which is the payments to the MUs is
$$C(\bm{x},\bm{p}) = \sum\limits_{i=1}^N p_i \cdot x_i.$$

SP's payoff characterizes the gap between utility and cost, which is formulated as
\begin{eqnarray}
    U(\bm{x},\bm{p}) &=& \phi(\bm{x}) - \sum\limits_{i=1}^N p_i \cdot x_i \label{eq_3-1-1}\\
                     &=& \phi(\bm{x}) - \bm{p}^T \bm{x}. \nonumber
\end{eqnarray}
The utility function for the SP is defined as follows
\begin{equation}\label{eq_3-1-2}
    \phi(\bm{x}) = \lambda \cdot \ln(1 + \sum\limits_{i=1}^N \ln(1+x_i)),
\end{equation}
where $\lambda$ is the SP's parameter specific to the SP's sensing task. The $ln(1+x_i)$ term reflects the SP's diminishing return on the service of MU $i$, and the outer $ln$ term reflects the SP's diminishing return on the number of the MUs. This kind of utility function has been widely accepted to represent the utility of the SP in the previous works \cite{yang2016incentive,maharjan2016optimal,zhan2018incentive}.

For analyzing conveniently in the following part, we set
\begin{equation}\label{eq_3-1-3}
    b = 1 + \sum\limits_{i=1}^N \ln(1 + x_i).
\end{equation}
Therefore, (\ref{eq_3-1-2}) can be rewritten as
$$\phi(\bm{x}) = g(b) = \lambda \cdot \ln(b).$$

\subsection{MUs' Payoff}\label{sec_3-2}
Every MU will carry a smartphone, and the MU's primary goal is to meet its own basic needs, which including making phone calls, sending messages, etc. Due to the development of microprocessors, nowdays, the smartphones are more and more powerful. After satisfying the MUs' basic needs, the smartphone still left some resources, which can be used to do some other activities, such as entertainment, noise monitoring, traffic monitoring, etc. That is, for MU $n$, the remaining energy of its smartphone can be used to entertain itself or participate in the MCS. Therefore, when deciding whether to participate in the MCS and distribute how many resources to the sensing task, each MU will consider both the resources demand of itself and rewards provided by the SP.

Let $\xi_n$ denote the MU $n$'s own resources demand. Due to the uncertainties of $n$'s behaviors\footnote{E.g. sometime $n$ may use the smartphone for long time entertainments, which leads $n$ to have few resources to participate in the crowdsensing.}, $\xi_n$ is defined as a random variable, also with a certain interval $[{\underline{\xi_n}}, \overline{\xi_n}]$. Where $\xi_n$ follows a probability distribution function $f_n(\xi_n)$ and a cumulative distribution function $F_n(\xi_n)$. Suppose that $n$ has $\tau_n$ units resources remaining, let $\delta_n$ denote the average revenue achieved from one unit of $n$'s own demand and $c_n$ denote the cost for one unit of $n$'s resource consumption. Then, MU $n$'s expected profit (from serving its own demand) is
\begin{eqnarray}
    R_n(\tau_n) \!\!\!\!\!&=&\!\!\!\!\! (\delta_n - c_n) \cdot E_{\xi_n}[min\{\xi_n, \tau_n\}] \label{eq_3-2-1} \\
                \!\!\!\!\!&=&\!\!\!\!\! (\delta_n - c_n) \cdot (\int\limits_{\underline{\xi_n}}^{\tau_n} {\xi_n f_n(\xi_n) d \xi_n} + \int\limits_{\tau_n}^{\overline{\xi_n}} {\tau_n f_n({\xi_n}) d \xi_n}). \nonumber
\end{eqnarray}

Now, if $MU$ $n$ admits $x_n$ units of resources to participate in the MCS for serving the SP, the resources left for serving $n$'s own demand are $\tau_n - x_n$, and a feasible $x_n$ must satisfy $x_n \le \tau_n$, obviously. For feasible $x_n$ and $p_n$, the MU $n$'s total profit including both the profit from serving its own demand and the profit from sensing for the SP is defined as
\begin{equation}\label{eq_3-2-2}
    R_n^{TOT}(x_n, p_n) = R_n(\tau_n - x_n) + p_n \cdot x_n - c_n \cdot x_n,
\end{equation}
where $p_n \cdot x_n - c_n \cdot x_n$ is the profit from serving the SP, including the sensing income (i.e. SP's payment $p_n \cdot x_n$) and the sensing cost $c_n \cdot x_n$.

Finally, the payoff of MU $n$ is the profit increment when providing sensing data for the SP, which is denoted by
\begin{eqnarray}
    U_n(x_n,p_n) \!\!\!\!\!&=&\!\!\!\!\! R_n^{TOT}(x_n,p_n) - R_n^{TOT}(0,0) \label{eq_3-2-3} \\
                 \!\!\!\!\!&=&\!\!\!\!\! R_n(\tau_n - x_n) - R_n(\tau_n) - c_n \cdot x_n + p_n \cdot x_n, \nonumber
\end{eqnarray}
where $R_n(\tau_n - x_n) - R_n(\tau_n) - c_n \cdot x_n$ is the MU $n$'s profit loss induced by sensing for the SP. Obviously, when the MU $n$ dose not participate in the MCS, its reservation payoff is $0$.

\subsection{Problem Formulation}\label{sec_3-3}
%We aim to design an incentive mechanism under which the SP will reward the MUs on the amounts of resources which are used for serving the SP. Therefore, the proposed mechanism can encourage MUs to participate in the crowdsensing and hence the sensing quality of the SP. In this paper, we formulate the incentive mechanism of crowdsensing as a Stackelberg game \cite{osborne1994course} based on non-cooperative game theory.
In order to encourage MUs to participate in the MCS and hence enhance the sensing quality of the SP, in this paper, the incentive mechanism of MCS is formulated as a Stackelberg game \cite{osborne1994course} based on non-cooperative game theory.
In the Stackelberg game, participants will be classified into two groups, namely leaders and followers, where the leaders have the privilege of moving first while the followers will move according to the leaders' actions.
Specifically, the MCS game studied in this paper is modeled as a single-leader with multi-followers Stackelberg game with two stage, where the SP acts as the leader and all the MUs act as the followers.
%Here, we formulate our work into a single-leader with multi-followers Stackelberg game, where the SP acts as the leaders, and all the MUs act as the followers. In the two-stage Stackelberg game,
Firstly, the SP (i.e. game leader) specifies the pricing strategy $\bm{p} = [p_1,p_2,\cdots,p_N]^T$. In the second stage, each MU acts as the game follower determine its sensing plan $x_n^*(p_n)$ to maximize its own payoff.

Given the definition of the Stackelberg game, the \textbf{SE} of the proposed incentive mechanism is defined as follows.
\begin{definition}
Let $\bm{p}^* = [p_1^*,p_2^*,\cdots,p_N^*]^T$ be the optimal solution to the SP and $x_n^*(p_n^*)$ be the optimal solution to the MU $n$. Then, $(\bm{p}^*,\bm{x}^*)$ is the \textbf{SE} for the proposed incentive mechanism if for any possible solution $(\bm{p},\bm{x})$, following conditions are satisfied:
\begin{eqnarray}
  U_n(\bm{x}_{-n}^*,x_n^*,\bm{p}) &\ge& U_n(\bm{x}_{-n}^*,x_n,\bm{p}), \nonumber \\
  U(\bm{x}(\bm{p}^*),\bm{p}_{-n}^*,p_n^*) &\ge& U(\bm{x}(\bm{p}_{-n}^*,p_n),\bm{p}_{-n}^*,p_n), \nonumber
\end{eqnarray}
where $\bm{x}_{-n}^*$ and $\bm{p}_{-n}^*$ indicate the MUs' strategy profile and SP's strategy profile excluding $n$th MU, respectively.
\end{definition}

\section{Incentive Mechanism for Static MCS Game}\label{sec_4}
In this section, how to design the incentive mechanism for the static MCS game by solving the Stackelberg game defined in Section~\ref{sec_3-3} is demonstrated. In the static MCS game, the main challenges are (a) how to develop the resource allocation strategy for the MUs and (b) how to develop a pricing strategy for the SP. In the following, it is firstly proven that for any feasible $\bm{p}=[p_1,p_2,\cdots,p_N]^T$, each MU has a unique optimal resource allocation strategy in the second stage (Section~\ref{sec_4-1}). Afterwards, the pricing strategy in the first stage is presented and the existence and uniqueness of the \textbf{SE} for the proposed incentive mechanism is proven (Section~\ref{sec_4-2}). Finally, summary of the overall static MCS game is provided (Section~\ref{sec_4-3}).

\subsection{Optimal Resource Allocation Strategy at MU}\label{sec_4-1}
%First, we analyze MU $n$'s optimal resource allocation strategy $x_n^*$ under SP's any feasible pricing strategy $p_n$ in the second stage.
In this subsection, the MU $n$'s optimal resource allocation strategy $x_n^*$ under SP's any feasible pricing strategy $p_n$ in
the second stage is analyzed. The optimal resource allocation strategy $x_n^*$ for MU $n$ is based on the following optimization problem
\begin{eqnarray}
  \max \!\!\! && \!\!\! U_n(x_n,p_n) \label{eq_4-1-1} \\
  s.t. \!\!\! && \!\!\! x_n \in [0,\tau_n]. \nonumber
\end{eqnarray}
\begin{lemma}\label{lemma_1}
Given any feasible $p_n$, MU $n$'s optimal resource allocation strategy $x_n^*$ satisfies
\begin{equation}\label{eq_4-1-2}
    x_n^*(p_n) = \left\{
\begin{array}{rcl}
0,       &      & if \;\; p_n \in [0,{\tilde p_n})\\
\tau_n - F_n^{(-1)}(\frac{\delta_n - p_n}{\delta_n - c_n}),     &      & if \;\; p_n \in [{\tilde p_n}, \delta_n]\\
\tau_n,     &      & if \;\; p_n \in (\delta_n, +\infty),
\end{array} \right.
\end{equation}
where $\tilde p_n = c_n + (\delta_n - c_n)[1-F_n(\tau_n)]$.
\end{lemma}
\begin{proof}
According to Eq.~(\ref{eq_3-2-1}) and (\ref{eq_3-2-3}), $U_n(\tau_n,x_n,p_n)$ (Hereafter, $U_n$ will be used instead for convenience) satisfies
\begin{eqnarray}
  U_n \!\!\!\!\!&=& \!\!\!\!\! (\delta_n - c_n)(\int\limits_{\underline{\xi_n}}^{\tau_n - x_n} {\xi_n f_n(\xi_n) d \xi_n} + \nonumber\\
    \!\!\!\!\! && \!\!\!\!\!  \int\limits_{\tau_n - x_n}^{\overline{\xi_n}} {(\tau_n - x_n) f_n(\xi_n) d \xi_n}) - c_n x_n + p_n x_n + R_n(\tau_n), \nonumber
    %\!\!\!\! && \!\!\!\!  c_n \cdot x_n + p_n \cdot x_n - R(\tau_n),
\end{eqnarray}
where
\begin{eqnarray}
  \!\!\!\!\! && \!\!\!\!\! \int\limits_{\underline{\xi_n}}^{\tau_n - x_n} {\xi_n f_n(\xi_n) d \xi_n} = \int\limits_{\underline{\xi_n}}^{\tau_n - x_n} {\xi_n d F_n(\xi_n)} \nonumber \\
  \!\!\!\!\! &=& \!\!\!\!\! (\tau_n - x_n)F_n(\tau_n - x_n) - \int\limits_{\underline{\xi_n}}^{\tau_n - x_n}{F_n(\xi_n) d \xi_n}, \nonumber
\end{eqnarray}
and
\begin{eqnarray}
  \int\limits_{\tau_n - x_n}^{\overline{\xi_n}} {(\tau_n - x_n) f_n(\xi_n) d \xi_n} \!\!\!\!\! &=& \!\!\!\!\! (\tau_n - x_n)(1-F_n(\tau_n - x_n)). \nonumber
\end{eqnarray}
Therefore,
\begin{eqnarray}
  U_n \!\!\!\!\! &=& \!\!\!\!\! (\delta_n - c_n)(\tau_n - x_n - \int\limits_{\underline{\xi_n}}^{\tau_n - x_n}{F_n(\xi_n) d\xi_n}) + \nonumber \\
  \!\!\!\!\! && \!\!\!\!\! (p_n - c_n) \cdot x_n - R_n(\tau_n). \nonumber
\end{eqnarray}
Then, the first- and second-order derivatives of $U_n$ with respect to $x_n$ can be derived as follows respectively
$$
\left\{ \begin{array}{l}
\frac{{\partial {U_n}}}{{\partial {x_n}}} = (\delta_n - c_n)(F_n(\tau_n - x_n) - 1) + p_n - c_n,\\
\frac{{{\partial ^2}{U_n}}}{{{\partial }{x_n}^2}} =  - (\delta_n - c_n) f_n(\tau_n - x_n).
\end{array} \right.
$$
Since $\delta_n \ge c_n$ and $f_n(\cdot) > 0$, it can be easily derived that the second-order derivative of $U_n$ to $x_n$ is negative. Also, the $x_n$ set of Problem~(\ref{eq_4-1-1}) is bounded and compact. Therefore, Problem~(\ref{eq_4-1-1}) is a strictly convex optimization problem. By setting the first-order of $U_n$ to $0$, the following equation is derived
\begin{equation}\label{eq_4-1-3}
    (\delta_n - c_n)(F_n(\tau_n - x_n) - 1) + p_n - c_n = 0.
\end{equation}
Let
$$
\tilde p_n = c_n + (\delta_n - c_n)(1 - F_n(\tau_n)).
$$
Due that $F_n(\cdot) \in [0,1]$, $\frac{{\partial {U_n}}}{{\partial {x_n}}} = 0$ is achievable only when $p_n \in [\tilde p_n, \delta_n]$. By solving~(\ref{eq_4-1-3}), we obtain
$$
x_n^*(p_n) = \tau_n - F_n^{(-1)}(\frac{\delta_n - p_n}{\delta_n - c_n}).
$$

Moreover, if $p_n < \tilde p_n$, then $\frac{{\partial {U_n}}}{{\partial {x_n}}}$ is less than $0$, thus the optimal allocation strategy for MU $n$ is $x_n^*=0$. And if $p_n > \delta_n$, $\frac{{\partial {U_n}}}{{\partial {x_n}}}$ is greater than $0$, thus the optimal allocation strategy for MU $n$ is $x_n^* = \tau_n$. This completes the proof.
\end{proof}

Note that when $p_n \in [\tilde p_n, \delta_n]$, the first-order derivative of $x_n^*$ with respect to $p_n$ is
$$
\frac{{\partial x_n^*}}{{\partial {p_n}}} = \frac{1}{{{f_n}(F_n^{( - 1)}(\frac{{{\delta _n} - {p_n}}}{{{\delta _n} - {c_n}}}))}}\frac{1}{{{\delta _n} - {c_n}}}.
$$
For any $p_n \in [\tilde p_n, \delta_n]$, it can be derived that $\frac{{\partial x_n^*}}{{\partial {p_n}}} > 0$. This indicates that the higher price offered by the SP is, the more resources allocated by MUs to the SP are. The second derivative of $x_n^*$ with respect to $p_n$ is
$$
\frac{{{\partial ^2}x_n^*}}{{\partial {p_n}^2}} = \frac{{{{f'}_n}(F_n^{( - 1)}(\frac{{{\delta _n} - {p_n}}}{{{\delta _n} - {c_n}}}))}}{{{{({f_n}(F_n^{( - 1)}(\frac{{{\delta _n} - {p_n}}}{{{\delta _n} - {c_n}}})))}^3}}}\frac{1}{{{{({\delta _n} - {c_n})}^3}}}.
$$
Hence, we can derive that if ${f'}_n(\cdot) \le 0$ (${f'}_n(\cdot) \ge 0$), the $x_n^*$ is a concave (convex) function of $p_n$. Notably, in this paper, it is assumed that ${f'}_n(\cdot) \le 0$\footnote{This is a very common hypothesis. There are lots of distributions satisfying that their probability distribution function are non-increasing, such as uniform distribution, exponential distribution, and Gaussian distribution, etc.}, thus $x_n^*$ is a concave function of $p_n$.
\subsection{Optimal Pricing Strategy at SP}\label{sec_4-2}
Now, the SP's optimal pricing strategy $\bm{p}^*$ in the first stage of the Stackelberg game is investigated. For any feasible $\bm{p}$, as has analyzed in Section~\ref{sec_4-1}, there is a unique $\bm{x}^*(\bm{p})$ indicates the MUs' best responses in the second stage of the Stackelberg game. The optimal pricing strategy at SP is determined by the following optimization problem
\begin{eqnarray}
  \max \!\!\!\! && \!\!\!\! U(\bm{x}^*,\bm{p}) \label{eq_4-2-1} \\
  s.t. \!\!\!\! && \!\!\!\! p_n \ge 0, \forall n \in \mathcal{N}, \nonumber \\
  var. \!\!\!\! && \!\!\!\! x_n^* \in [0, \tau_n]. \nonumber
\end{eqnarray}
For convenience, in the following part, $U$ will be utilized to replace $U(\bm{x}^*,\bm{p})$.

Since $x_n^*$ is the function of $p_n$, according to Eq.~(\ref{eq_3-1-1}), the first-order derivative of $U(\bm{x}^*,\bm{p})$ with respect to $p_n$ is
\begin{equation}\label{eq_4-2-2}
    \frac{{\partial U}}{{\partial {p_n}}} = \frac{{\partial g(b)}}{{\partial b}}\frac{1}{{1 + x_n^*}}\frac{{\partial x_n^*}}{{\partial {p_n}}} - {p_n}\frac{{\partial x_n^*}}{{\partial {p_n}}} - x_n^*.
\end{equation}
According to Lemma~\ref{lemma_1}, if $p_n < \tilde p_n$ or $p_n > \delta_n$, then $\frac{\partial x_n^*}{\partial p_n} \equiv 0$. Thus, if $p_n < \tilde p_n$, $\frac{\partial U}{\partial p_n} = 0$. This indicates that any $p_n \in [0, \tilde p_n)$ is indifferent to the SP. If $p_n > \delta_n$, $\frac{\partial U}{\partial p_n} = -\tau_n$. This indicates that any price $p_n$ greater than $\delta_n$ is dominated by $\delta_n$. Therefore, the analysis of $p_n \in [0, \tilde p_n]$ and $p_n > \delta_n$ is meaningless. The following parts only focus on $p_n \in [\tilde p_n, \delta_n]$, $\forall n \in \mathcal{N}$. For convenience, $g'(b)$ and $g''(b)$ are utilized to denote $\frac{\partial g(b)}{\partial b}$ and $\frac{\partial^2 g(b)}{\partial b^2}$ respectively.

\begin{lemma}\label{lemma_2}
For the optimal pricing strategy profile $\bm{p}^* = [p_1^*,p_2^*,\cdots,p_N^*]^T$, $p_n^*$ must satisfy
\begin{equation}\label{eq_4-2-3}
    p_n^* \le \frac{g'(b)}{1+x_n^*}.
\end{equation}
and in addition, $p_n^*$ must not blow $\tilde p_n$, or it is indifferent to the SP.
\end{lemma}
\begin{proof}
Assume that $p_n^* > \frac{g'(b)}{1+x_n^*}$, take it into Eq.~(\ref{eq_4-2-2}). Since $\frac{\partial x_n^*}{\partial p_n^*} > 0$, we can obtain that
$$\frac{{\partial U}}{{\partial {p_n^*}}} = (\frac{g'(b)}{1+x_n^*} - p_n^*) \frac{\partial x_n^*}{\partial p_n^*} - x_n^* < 0.$$
This implies that there exists a price $p_n < p_n^* - \varepsilon$ ($\varepsilon $ is a very small positive real number), which leads the SP to a higher payoff. This is in contradiction with that $p_n^*$ is the optimal pricing strategy. Therefore, the assumption does not hold. This completes the proof.
\end{proof}

Moreover, Lemma~\ref{lemma_2} can be interpreted from a physical perspective. $\frac{\partial \phi(\bm{x}^*)}{\partial x_n^*}$ is the SP's utility increasing due to the unit resource contribution of MU $n$, which is equal to $\frac{g'(b)}{1+x_n^*}$. While $p_n^*$ is the payoff which SP pays to the MU $n$ for its unit resource contribution. Obviously, SP's payoff increasing due to MU $n$'s contribution must greater than the payoff it pays to MU $n$, or it will not recruit MU $n$.
\begin{lemma}\label{lemma_3}
Optimization problem~(\ref{eq_4-2-1}) has an unique solution.
\end{lemma}
\begin{proof}
The Hassian matrix of $U$ is defined as $H$, which satisfies
$$
H = \left[ {\begin{array}{*{20}{c}}
{\frac{{{\partial ^2}U}}{{\partial {p_1}^2}}}&{\frac{{{\partial ^2}U}}{{\partial {p_1}\partial {p_2}}}}& \cdots &{\frac{{{\partial ^2}U}}{{\partial {p_1}\partial {p_N}}}}\\
{\frac{{{\partial ^2}U}}{{\partial {p_2}\partial {p_1}}}}&{\frac{{{\partial ^2}U}}{{\partial {p_2}^2}}}& \cdots &{\frac{{{\partial ^2}U}}{{\partial {p_2}\partial {p_N}}}}\\
 \vdots & \vdots & \ddots & \vdots \\
{\frac{{{\partial ^2}U}}{{\partial {p_N}\partial {p_1}}}}&{\frac{{{\partial ^2}U}}{{\partial {p_N}\partial {p_2}}}}& \cdots &{\frac{{{\partial ^2}U}}{{\partial {p_N}^2}}}
\end{array}} \right].
$$

According to Eq.~\ref{eq_3-1-1}, the second-order derivative of $U$ with respect to $p_n$ is
\begin{eqnarray}
    \frac{{{\partial ^2}U}}{{\partial {p_n}^2}} \!\!\!\! &=& \!\!\!\! \frac{{g''(b) - g'(b)}}{{{{(1 + x_n^*)}^2}}}{(\frac{{\partial x_n^*}}{{\partial {p_n}}})^2} - 2\frac{{\partial x_n^*}}{{\partial {p_n}}} + \label{eq_4-2-4} \\
    \!\!\!\! && \!\!\!\! (\frac{{g'(b)}}{{1 + x_n^*}} - {p_n})\frac{{{\partial ^2}x_n^*}}{{\partial {p_n}^2}}. \nonumber
\end{eqnarray}
Moreover, the second-order partial derivative of $U$ with respect to $p_i$ and $p_j$ is
\begin{equation}\label{eq_4-2-5}
    \frac{{\partial ^2}U}{\partial p_i \partial p_j} = \frac{{\partial ^2}U}{\partial p_j \partial p_i} = g''(b) \frac{1}{(1+x_i^*)(1+x_j^*)} \frac{\partial x_i^*}{\partial p_i} \frac{\partial x_j^*}{\partial p_j}.
\end{equation}
Set
\begin{equation}\label{eq_4-2-6}
    H_1 = \left[ {\begin{array}{*{20}{c}}
{{\lambda _1}}&0& \cdots &0\\
0&{{\lambda _2}}& \cdots &0\\
 \vdots & \vdots & \ddots & \vdots \\
0&0& \cdots &{{\lambda _N}}
\end{array}} \right],
\end{equation}
where $\lambda_n = (\frac{g'(b)}{1+x_n^*} - p_n)\frac{\partial^2 x_n^*}{\partial {p_n}^2} - 2 \frac{\partial x_n^*}{\partial p_n} - \frac{g'(b)}{(1+x_n^*)^2} (\frac{\partial x_n^*}{\partial p_n})^2$, $\forall n \in \mathcal{N}$. According to Lemma~\ref{lemma_2}, $\frac{g'(b)}{1+x_n^*} - p_n \ge 0$. Also, we have declared in Section~\ref{sec_4-1}, $\frac{\partial x_n^*}{\partial p_n} > 0$ and $\frac{\partial^2 x_n^*}{\partial {p_n}^2} \le 0$. Meanwhile, we can easily derive that $g'(b) > 0$. As a result,
$$\lambda_n \le 0.$$
Furthermore, set
$$
    H_2 = g''(b) \left[ {\begin{array}{*{20}{c}}
{{H_2}(1,1)}&{{H_2}(1,2)}& \cdots &{{H_2}(1,N)}\\
{{H_2}(2,1)}&{{H_2}(2,2)}& \cdots &{{H_2}(2,N)}\\
 \vdots & \vdots & \ddots & \vdots \\
{{H_2}(N,1)}&{{H_2}(N,2)}& \cdots &{{H_2}(N,N)}
\end{array}} \right],
$$
where $H_2(i,j) = H_2(j,i) = \frac{1}{(1+x_i^*)(1+x_j^*)}\frac{\partial x_i^*}{\partial p_i} \frac{\partial x_j^*}{\partial p_j}$, $\forall i,j \in \mathcal{N}$. Therefore, we can rewrite $H_2$ as
\begin{equation}\label{eq_4-2-7}
H_2 = g''(b)\bm{q}\bm{q}^T,
\end{equation}
where $\bm{q} = [q_1,q_2,\cdots,q_N]^T$, and $q_n = \frac{1}{1+x_n^*} \frac{\partial x_n^*}{\partial p_n}$. According to the definition of Hassian matrix, we can obtain that
$$H = H_1 + H_2.$$
Randomly select a vector $\bm{v} = [v_1,v_2,\cdots,v_N]^T$, where $v_i \in \mathbb{R}$ and the elements in $\bm{v}$ are not all $0$. Then we have that
$$
\bm{v}^T H \bm{v} = \bm{v}^T H_1 \bm{v} + \bm{v}^T H_2 \bm{v}.
$$
According to Eq.~(\ref{eq_4-2-6}), we can derive that
$$
\bm{v}^T H_1 \bm{v} = \sum\limits_{i=1}^N \lambda_k v_k^2 \le 0.
$$
Based on Eq.~(\ref{eq_4-2-7}), we have that
$$\bm{v}^T H_2 \bm{v} = g''(b)\bm{v}^T \bm{q} \bm{q}^T \bm{v} = g''(b)(\sum\limits_{i=k}^N{\frac{v_k}{1+x_k^*} \frac{\partial x_k^*}{\partial p_k}})^2.$$
Since $g''(b) = -\frac{1}{(1+b)^2} < 0$ and $\frac{\partial x_k^*}{\partial p_k} > 0$, we can derive that $\bm{v}^T H_2 \bm{v} < 0$. Therefore, we have that
$$\bm{v}^T H \bm{v} < 0.$$
This indicates that $U$ is a strictly concave function. Furthermore, the constraint set of Problem~(\ref{eq_4-2-1}) is nonempty, compact, and convex. Thus, Problem~(\ref{eq_4-2-1}) has a unique solution \cite{boyd2004convex}. This completes the proof.
\end{proof}

Through Lemma~\ref{lemma_1} and Lemma~\ref{lemma_3}, we can obtain that for the static MCS game in the first stage, SP has a unique pricing strategy profile $\bm{p}^*$, which satisfies
$$I: \;\;\;\; p_n^* = \mathop {\arg \max }\limits_{{p_n}} U(\bm{x}^*,\bm{p}).$$
In the second stage, each MU has a unique resource allocation strategy $x_n^*$, which satisfies
$$II: \;\;\;\; x_n^* = \mathop {\arg \max }\limits_{{x_n}} U_n(\tau_n, x_n, p_n^*).$$
\begin{theorem}\label{theorem}
There exists a unique \textbf{SE} in the static MCS game.
\end{theorem}

\begin{figure}[!ht]
  \centering
  \includegraphics[width=9cm]{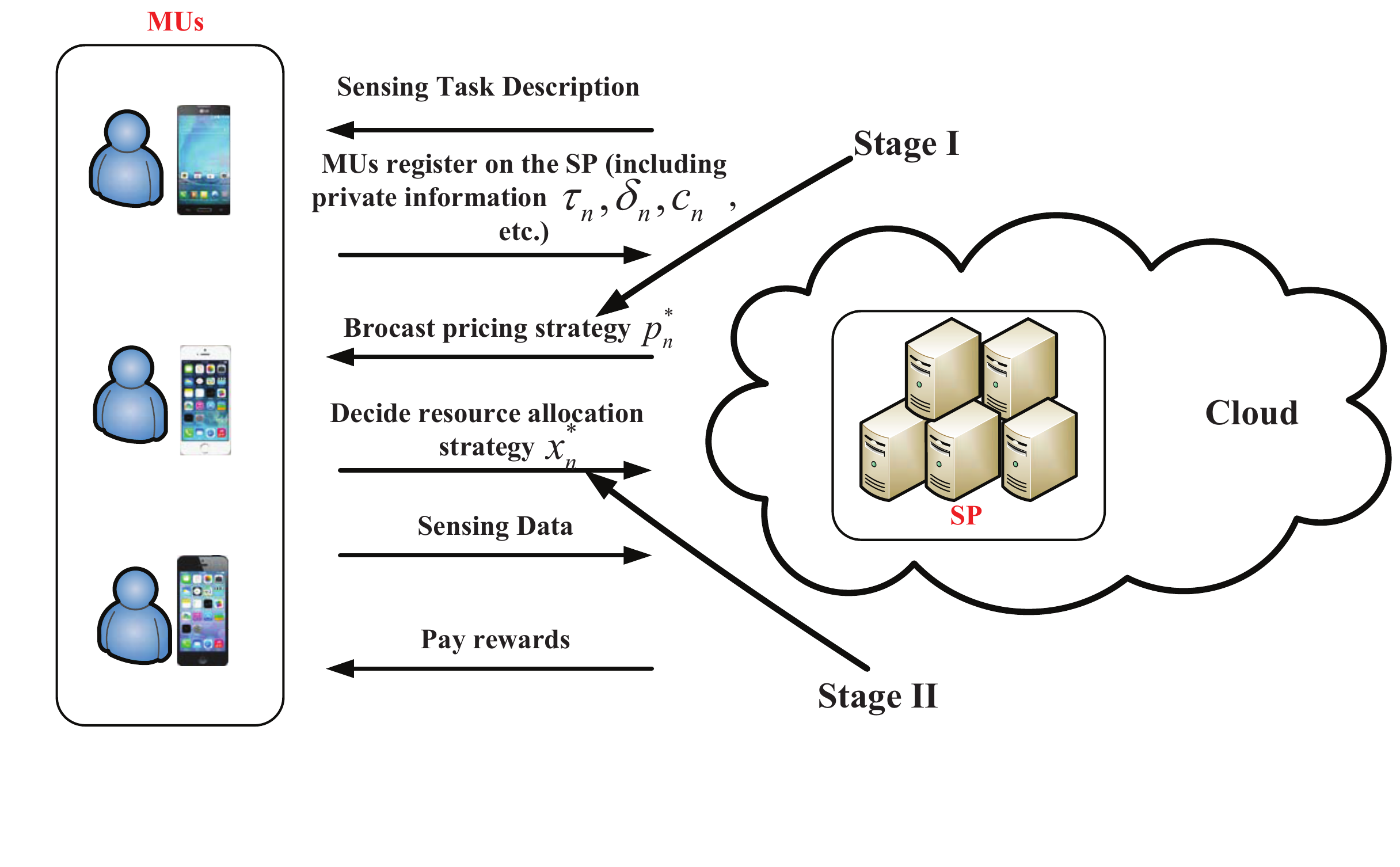}\\
  \caption{Interactions between MUs and SP in static MCS.}\label{fig_4-3}
\end{figure}

\subsection{Overall Static Crowdsensing Game}\label{sec_4-3}
In this subsection, the overall incentive mechanism based on static MCS game is presented. In the MCS system, a centralized SP where the task initiators could reside will lead the price decision and task allocation procedure. Fig.~\ref{fig_4-3} shows the detailed interactions between MUs and SP. First, the SP issues the sensing task, when the MUs are considering to joint the MCS, they need to register on the SP. Then the SP decides the pricing strategy according to the optimization problem defined by~(\ref{eq_4-2-1}). Optimization problem~(\ref{eq_4-2-1}) can be solved by Karush-Kuhn-Tucker (KKT) method \cite{boyd2004convex}. The MUs will be informed of the detailed pricing strategy, and then they could determine the optimal resource allocation strategy based on Lemma~\ref{lemma_1}. Finally, after the MUs completing the sensing task and uploading the sensing data, the SP will pay them the corresponding rewards.
%The detailed Algorithm of static crowdsensing game is shown in Algorithm~\ref{algorithm-1}.
%
%\begin{algorithm}[htb]
%\caption{The detailed algorithm of static crowdsensing game.}
%\label{algorithm-1}
%\textbf{Input:}
%the collector $s$, the requestor $d$, $\mathcal{R}=\{R_1,R_2,\cdots,R_k\}$, $\lambda_{ij}$ for any $i$ and $j$, the collector's bargaining power with respect to each relay candidate $\mathcal{B}_1=\{\alpha_{sR_1}^m, \alpha_{sR_2}^m, \cdots, \alpha_{sR_k}^m\}$, and the bargaining power of each relay candidate with respect to the collector $\mathcal{B}_2=\{\alpha_{R_1s}^m, \alpha_{R_2s}^m, \cdots, \alpha_{R_ks}^m\}$.\\
%\textbf{Output:} the selected relay $OR$ and its expected online rewards for $s$.
%\begin{algorithmic}[1]
%%Ensemble of classifiers on the current batch, $E_n$;
%\State \textbf{when the $s$ meets relay candidate $R_i$ do}
%\State $(OR,EFR_{max}^m)=NBF(s,d,\mathcal{R}\setminus\{R_i\},\lambda,\mathcal{B}_1,\mathcal{B}_2)$
%\State $A_{R_id}^m =  \int\limits_0^{T '_m}  {({{\eta '}_m} - {\beta _m}t) \cdot {f_{{X_{{R_i}d}}}}(t)dt}$
%\State $\pi_m = A_{R_id}^m$
%\State $ENR_{sR_i}^m = \frac{\alpha_{sR_i}^m \pi^m - \alpha_{sR_i}^m r_{R_i}^m + \alpha_{R_is}^m t_s^m }{\alpha_{sR_i}^m + \alpha_{R_is}^m} - t_s^m$
%\If {$ENR_{sR_i}^m \ge EFR_{max}^m$} \\
%$\quad$\Return $R_i$ and $ENR_{sR_i}^m$;
%\Else \\
%$\quad$\Return $OR$ and $EFR_{max}^m$;
%\EndIf
%\end{algorithmic}
%\end{algorithm}

\section{Dynamic Incentive Mechanism (DIM) Design for MCS}\label{sec_5}
In this section, a dynamic incentive mechanism (DIM) based on deep reinforcement learning (DRL) approach is designed for MCS. Since solving optimization problem $\rm{OP}_{\rm{TI}}$ in Eq.~(\ref{eq_4-2-1}) directly requires MUs' private information such as $\tau_n,\delta_n,c_n$, having them can be impractical and unsafe for MUs in reality. Hence, a DRL approach is designed to learn the optimal strategy directly from game history, during which no prior knowledge about MUs is required. In the following, we firstly establish the dynamic MCS game as a Markov Decision Process (MDP) for dynamic MCS game. Then, we present the DRL algorithm designed for SP to find the optimal pricing strategy in Section~\ref{sec_4-2}. Finally, we demonstrate the overall dynamic MCS game based on DRL.

\subsection{\emph{MDP} for Dynamic MCS Game}\label{sec_5-1}
The MDP ($\mathcal{M}$) for dynamic MCS game is composed of state space ($\mathcal{S}$), action space ($\mathcal{A}$), state transition probability function ($\mathcal{P}$), and reward function ($\mathcal{R}$), namely $\mathcal{M}=<\mathcal{S},\mathcal{A},\mathcal{P},\mathcal{R},\gamma>$ shown in Fig.~\ref{fig_5-1}.
\begin{figure}[!t]
\centering
\includegraphics[width=0.8\linewidth]{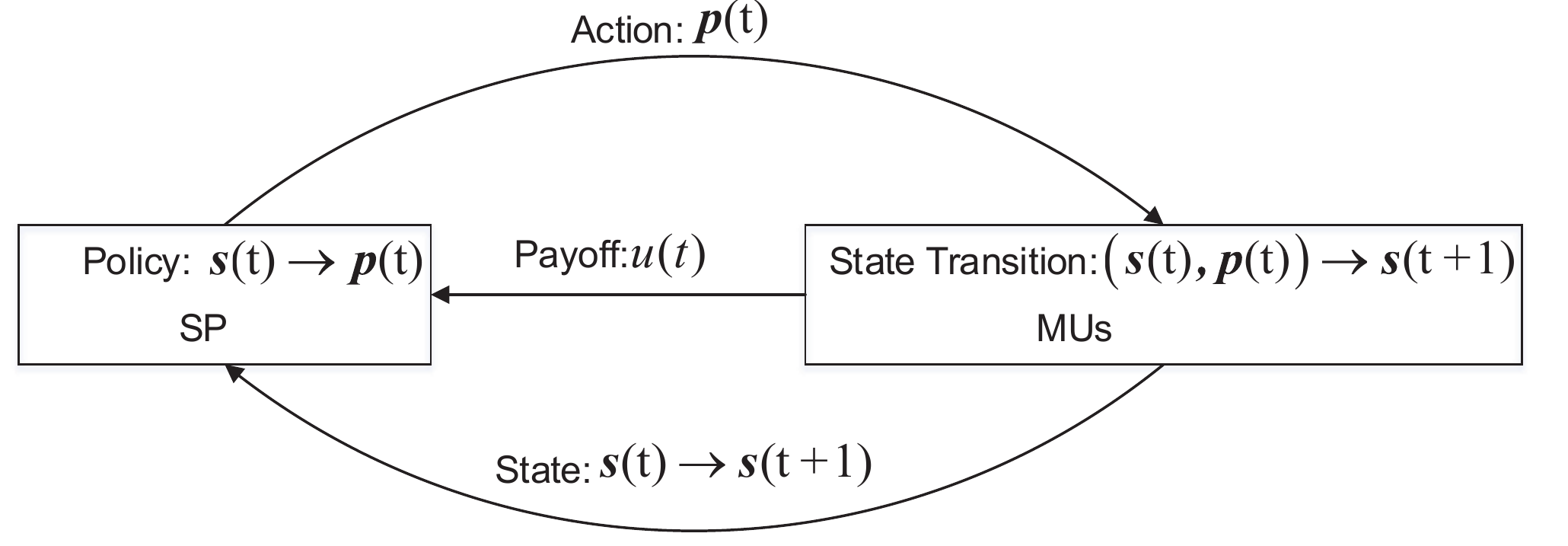}
\caption{Markov decision process.}\label{fig_5-1}
\end{figure}
\subsubsection{State space}
We define the state space of $\mathcal{M}$ as $\mathcal{S} =\{\bm{s}(t)|\forall~t \in \mathbb{N}\}$, where $\bm{s}(t)=[\bm{p}(t-L),\bm{x}(t-L),\cdots,\bm{p}(t-1),\bm{x}(t-1)] \in \mathbb{R}^{N \times 2L}$ denotes the past $L$ times game history between SP and MUs. More precisely, $\bm{p}(t)$ is SP's pricing strategy profile at step $t$ and $\bm{x}(t)$ is MUs' resource allocation strategy at step $t$. When $t \leq L$, $\bm{p}(t-L)$ and $\bm{x}(t-L)$ can be initialized randomly.
\subsubsection{Action space}
The action space of $\mathcal{M}$ is denoted as $\mathcal{A} =\{\bm{p}(t)|\forall~t \in \mathbb{N}\}$, which consists of SP's pricing strategy profiles.
\subsubsection{State transition probability function}
The state transition probability function is defined as  $\mathcal{P}: \mathcal{S}\times\mathcal{A}\times\mathcal{S}\rightarrow [0,1]$, meaning that the state $\bm{s}(t)$ will transit into $\bm{s}(t+1)$ satisfying $\bm{s}(t+1)\sim P(\bm{s}|\bm{s}(t),\bm{p}(t))$ after taking action $\bm{p}(t)$.
\subsubsection{Reward function}
The reward function $\mathcal{S}\times\mathcal{A}\rightarrow\mathbb{R}$ is proportional to the SP's payoff as follows
\begin{equation}
\label{eq.4-0}
r(t) = \xi U(\bm{x}(t),\bm{p}(t)),
\end{equation}
where $\xi$ is the scaling factor.
\subsubsection{Problem formulation}
Define SP's policy as $\pi:\mathcal{S}\times\mathcal{A}\rightarrow [0,1]$.
%At step $t$, the SP will choose it pricing strategy $\bm{p}(t)=\bm{a}(t)$ satisfying $\bm{a}(t)\sim\pi(\bm{a}|\bm{s}(t))$ and broadcast it to MUs. After that, MUs will determine their optimal resource allocation strategy $\bm{x}_{t}$ based on Lemma~\ref{lemma_1} and the SP's state $\bm{s}(t)$ will be updated into $\bm{s}(t+1)$. Finally, the SP will obtain its payoff $r(t)$.
Then, the goal of DRL-based DIM design is to find the optimal policy for SP satisfying
\begin{equation}
\begin{aligned}
\label{eq.4-1}
\bm{\theta}^{*} &= \arg\max_{\bm{\theta}}~L_1(\pi_{\bm{\theta}})
= \arg\max_{\bm{\theta}}~\int_{\mathcal{S}}\rho(\bm{s})V\big(\bm{s}\big)d\bm{s}\\
&= \arg\max_{\bm{\theta}}~\int_{\mathcal{S}}\rho(\bm{s})\int_{\mathcal{A}}\pi_{\bm{\theta}}(\bm{p}|\bm{s})Q\big(\bm{s},\bm{p}\big)d\bm{p}d\bm{s}
\end{aligned}
\end{equation}
%\mathbb{E}_{\bm{s}_{m}\sim\rho_{m}(\bm{s})}\Big[V_m\big(\bm{s}_{m}\big)\Big]\bm{p}_m
where
$V\big(\bm{s}\big)=\mathbb{E}\Big[\sum_{l=1}^{\infty}\gamma^{l-1}r(l)|\bm{s}(1)=\bm{s}\Big]$ is state value function, $Q(\bm{s},\bm{p})=\mathbb{E}\big[\sum_{l=1}^{\infty}\gamma^{l-1}r(l)|\bm{s}(1)=\bm{s},\bm{p}(1)=\bm{p}\big]$
is the action value function, $\rho(\bm{s})=\int_{\mathcal{S}}\sum_{l=1}^{\infty}\gamma^{l-1}P\big(\bm{s}(l)=\bm{s}|\bm{s}(1), \pi_{\bm{\theta}}\big)d\bm{s}(1)$ is the state probability distribution, and $\gamma \in [0,1]$ is a discount factor.

\subsection{Policy Optimization for MDP}
We adopt policy gradient method based on the proven actor-critic framework \cite{li2017deep} to deal with policy optimization problem described in Eq.~(\ref{eq.4-1}). Specifically, we employ an actor network $\pi_{\bm{\theta}}$ parameterized by $\bm{\theta}$ to generate stochastic action $\bm{a}(t) \sim \pi_{\bm{\theta}}(\cdot|\bm{s}(t))$ and a critic network $V_{\bm{\omega}}$ parameterized by $\bm{\omega}$ to approximate the state value function $V \big(\bm{s}(t)\big)$.

%[DRL2017] Li, Y.: `Deep reinforcement learning: An overview', arXiv preprint arXiv:1701.07274, 2017.
%[sutton2000policy] R. S. Sutton, D. A. McAllester, S. P. Singh, and Y. Mansour, ¡°Policy gradient methods for reinforcement learning with function approximation,¡± in NIPS'00, 2000, pp.1057¨C1063.
Referring to the stochastic policy gradient theorems in \cite{sutton2000policy} and \cite{schulman2015trust}, the policy gradient can be calculated as

\begin{equation}
\begin{aligned}
\nabla_{\bm{\theta}}L_1(\bm{\theta})&=\mathbb{E}_{\bm{s}\sim\rho,\bm{p}\sim\pi_{\bm{\theta}}}\Big[\nabla_{\bm{\theta}}\log\pi_{\bm{\theta}}(\bm{p}|\bm{s})Q\big(\bm{s},\bm{p}\big)\Big]\\
&=\mathbb{E}_{\bm{s}\sim\rho,\bm{p}\sim\pi_{\bm{\hat\theta}}}\Big[\nabla_{\bm{\theta}}\log\pi_{\bm{\theta}}(\bm{p}|\bm{s})f\big(\bm{s},\bm{p}\big)A\big(\bm{s},\bm{p})\big)\Big]\\
\end{aligned}
\end{equation}
%[nips2000] Richard S. Sutton, David McAllester, Satinder Singh, et al.Policy Gradient Methods for Reinforcement Learning with Function Approximation. Advances in Neural Information Processing Systems 12, pp. 1057{1063, MIT Press, 2000
where $f(\bm{s},\bm{p})=\frac{\pi_{\bm{\theta}}(\bm{p}|\bm{s})}{\pi_{\bm{\hat\theta}}(\bm{p}|\bm{s})}$, $A\big(\bm{s},\bm{p}\big)=Q\big(\bm{s},\bm{p}\big)-V\big(\bm{s}\big)$ is the advantage function, and the parameter of policy for sampling $\bm{p}$ is $\bm{\hat\theta}$.

Furthermore, in order to increase stability of training process based on policy gradient, \cite{schulman2017proximal} proposed proximal policy optimization (PPO) method, which clips the policy gradient as
\begin{equation}
\begin{aligned}
\label{eq.4-9}
\nabla_{\bm{\theta}}L_1^{'}(\bm{\theta})&=\nabla_{\bm{\theta}}\mathbb{E}_{\bm{s}\sim\rho,\bm{p}\sim\pi_{\bm{\hat\theta}}}\Big[\min\big(f(\cdot)A(\cdot), \eta(f(\cdot))A(\cdot)\big)\Big]\\
&\approx\sum_{k=1}^{D}\nabla_{\bm{\theta}}\log\pi_{\bm{\theta}}(k)\min\big[f(k)\hat A(k), \eta(f(k))\hat A(k)\big],
\end{aligned}
\end{equation}
where $\nabla_{\bm{\theta}}\log\pi_{\bm{\theta}}(k)=\nabla_{\bm{\theta}}\log\pi_{\bm{\theta}}(\bm{p}(k)|\bm{s}(k))$, $f(k) = \frac{\pi_{\bm{\theta}}(\bm{p}(k)|\bm{s}(k))}{\pi_{\bm{\hat\theta}}(\bm{p}(k)|\bm{s}(k))}$, $\hat A(k) = \sum_{l=k}^{D}r(l)+V_{\bm{\omega}}(\bm{s}(D+1))-V_{\bm{\omega}}(\bm{s}(k))$, $D$ is number of samples for policy gradient estimation at each training step, and $\eta(x)$ is the piecewise function with intervals $[x<1-\varepsilon, 1-\varepsilon\leq x \leq 1+\varepsilon, x>1+\varepsilon]$, $\varepsilon$ is an adjustable parameter.

Finally, the loss function for optimizing the critic network $V_{\bm{\omega}}$ is defined as
\begin{equation}
\begin{aligned}
\label{eq.4-10}
L_2(\bm{\omega}) &= \mathbb{E}_{\bm{s}\sim\rho(\bm{s})}\Big[-V_{\bm{\omega}}(\bm{s}) +\mathbb{E}_{\bm{s}'\sim P, \bm{p}\sim \pi_{\bm{\hat\theta}}}\big[r+V_{\bm{\omega}}(\bm{s}')\big]\Big]^2\\
&\approx\sum_{k=1}^{{D}}\Big[-V_{\bm{\omega}}\big(\bm{s}(k)\big)+\sum_{l=k}^{D}r(l)+V_{\bm{\omega}}\big(\bm{s}({D+1})\big)\Big]^2.
\end{aligned}
\end{equation}

\subsection{Proposed DRL-based DIM for SP}\label{sec_5-2}
\subsubsection{Procedure of dynamic game}
As illustrated in Fig.~\ref{fig_4-4}, the SP issues the sensing task firstly. Then, MUs register on the SP if they determine to join the MCS. At game step $t$, the SP will decide the pricing strategy $\bm{p}(t)$ according to its game memory matrix $\bm{s}(t)$. After that, MUs will obtain the detailed pricing strategy $\bm{p}(t)$ and then determine the optimal resource allocation strategy $\bm{x}(t)$ based on Lemma~\ref{lemma_1}. After MUs completing the sensing task and uploading the sensing data, the SP will pay them the corresponding rewards and attain its own payoff $r(t)$. Finally, the SP will update its negotation history into $\bm{s}(t+1)$ and start the new game.
\begin{figure*}[!htb]
  \centering
  \includegraphics[width=14cm]{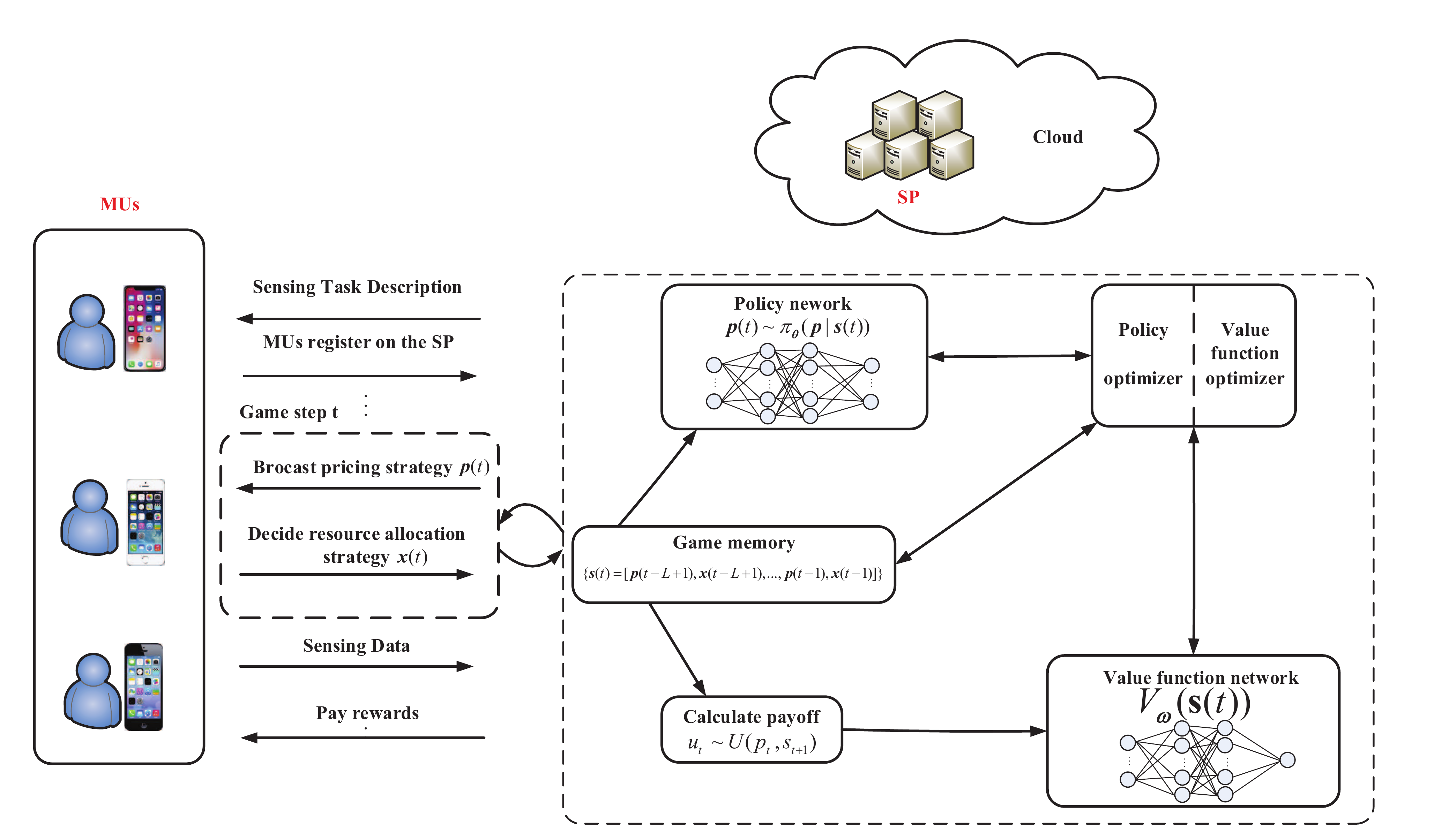}\\
  \caption{Dynamic mobile crowdsensing game.}\label{fig_4-4}
\end{figure*}
\subsubsection{Update actor and critic networks}
Each time after $D$ times dynamic mobile crowdsensing game, the actor and critic networks will be updated. More specifically, the SP will firstly calculate $V_{\bm{\omega}}\big(\bm{s}(k)\big)~(k=1, \cdots, D)$ by critic network. Afterwards, it will count $\sum_{l=k}^{D}r(l)$, $f(k)$, and $\hat A(k)~(k=1, \cdots, D)$. Then, the actor network $\pi_{\bm{\theta}}$ can be updated through gradient ascend method as
\begin{equation}
\label{eq4.11}
\bm{\theta} \leftarrow \bm{\theta} - l_1\nabla_{\bm{\theta}}L_1^{'}(\bm{\theta}),
\end{equation}
where $l_1$ is the learning rate for actor network updation. The critic network $V_{\bm{\omega}}$ can be updated through gradient descend method as
\begin{equation}
\label{eq4.12}
\bm{\omega} \leftarrow \bm{\omega} + l_2\nabla_{\bm{\omega}}L_2(\bm{\omega}),
\end{equation}
where $l_2$ is the learning rate for critic network updation.

\subsubsection{Detailed explanation of DRL-based DIM}
Algorithm \ref{algor-3-2} illustrates the pseudocode for the proposed DRL-based DIM. When a game begins, the SP initializes its state (Line 1). At each step $t$ in dynamic mobile crowdsensing game, by taking its state as the input of its policy network $\pi_{\bm{\theta}}$, the SP's pricing strategy profile $\bm{p}(t)$ can be determined and sent to MUs (Line 8). After obtaining MUs' resouce allocation (Line 9), the SP can calculate its payoffs and reward (Line 10). Then, the SP will update its state and record its game information (Line 11-12). Parameters of both actor network and critic network are optimized every $D$ step by utilizing the past $D$ game record (Line 10-11). After updating these two neural networks based on gradient ascent method and gradient descent method respectively for $M$ times (Line 15-16), a new episode of dynamic game will start (Line 2) and the SP will clear its replay buffer (Line 3).
\begin{algorithm}[!t]
\renewcommand{\algorithmicrequire}{\textbf{Input:}}
\renewcommand\algorithmicensure {\textbf{Output:}}
\caption{DRL-based DIM for SP}
\label{algor-3-2}
\begin{algorithmic}[1]
\REQUIRE Game history of the SP $\bm{s}(t)$.
\ENSURE Pricing strategy profile of the SP $\bm{p}(t)$.
\STATE Initialize $\bm{s}(t)$, $\pi_{\bm{\theta}}$ and $V_{\bm{\omega}}$.
\FOR {Episode in $1,2,\cdots$}
\STATE Clear the replay buffer $\mathcal{D}$.
\FOR {Step $t$ in $1,2,\cdots,D$}
\IF {Episode $\geq$ 2}
\STATE $\bm{s}(1) = \bm{s}(D+1)$
\ENDIF
\STATE Input $\bm{s}(t)$ into the policy network $\pi_{\bm{\theta}}$ and derive the pricing strategy $\bm{p}(t)$.
\STATE Receive the MUs' resource allocation strategy $\bm{x}(t)$.
\STATE Calculate the SP's payoff $U(\bm{x}(t),\bm{p}(t))$ by (\ref{eq_3-1-1}) and its reward $r_t$ by (\ref{eq.4-0}).
\STATE Update state $\bm{s}(t)$ into state $\bm{s}(t+1)$.
\STATE Store  $\{\bm{s}(t),\bm{p}(t),\bm{s}(t+1),r(t)\}$ into $\mathcal{D}$.
\ENDFOR
\FOR {i in $1,2,\cdots, M$}
\STATE Calculate $\nabla_{\bm{\theta}}L_1^{'}(\bm{\theta})$ and $\nabla_{\bm{\omega}}L_2(\bm{\omega})$ via (\ref{eq.4-9}) and (\ref{eq.4-10}).
\STATE Update $\bm{\theta}$ and $\bm{\omega}$ through (\ref{eq4.11}) and (\ref{eq4.12}).
\ENDFOR
\ENDFOR
\end{algorithmic}
\end{algorithm}
%\textbf{Yu: what is $epoch$ and  $maximal\_epoch$ in Algorithm 1?}

\section{Numerical Results}\label{sec_6}
In this section, numerical simulations are conducted. Specifically, $5$ MUs are randomly generated. For each MU, $c_n$ and $\delta_n$ are randomly from $[0,1]$ while guaranteeing $\delta_n > c_n$. We set the total available resources $\tau_n$ of each MU to $20$ units, and randomly drawn the own resources demand $\xi_n$ from a uniform distribution in $[0,25]$.

\subsection{\textbf{SE} under Dynamic Crowdsensing Game}\label{sec_6-1}
In this part, simulations have been performed to evaluate the system performance of the dynamic crowdsensing game, with $\lambda = 50$, $c_n$, $\delta_n$ randomly select from $[0,1]$ and $\delta_n > c_n$. Fig.~\ref{fig_6-2-1-1} and~\ref{fig_6-2-1-2} show the DRL based pricing strategy and resource allocation strategy converge to the optimal policy quickly in the dynamic crowdsensing game, which matches the theoretical results of the \textbf{SE} given in the static crowdsensing game in Section~\ref{sec_4}. This indicates that the incentive mechanism based on our designed DRL method will efficiently work. Meanwhile, as shown in Fig.~\ref{fig_6-2-2-1} and~\ref{fig_6-2-2-2}, the SP's payoff and MUs' payoff will also converge to the stable state quickly. More specifically, we can see that DIM outperforms the greedy and random pricing strategies. This is because in greedy and random policies, the SP always issues the higher prices to the MUs which leads the MUs obtain higher payoffs as shown in Fig.~\ref{fig_6-2-2-2} and SP obtains lower payoff.
\begin{figure*}[!ht]
\centering
\subfloat[Optimal pricing strategy.]{\includegraphics[width = 0.25\linewidth]{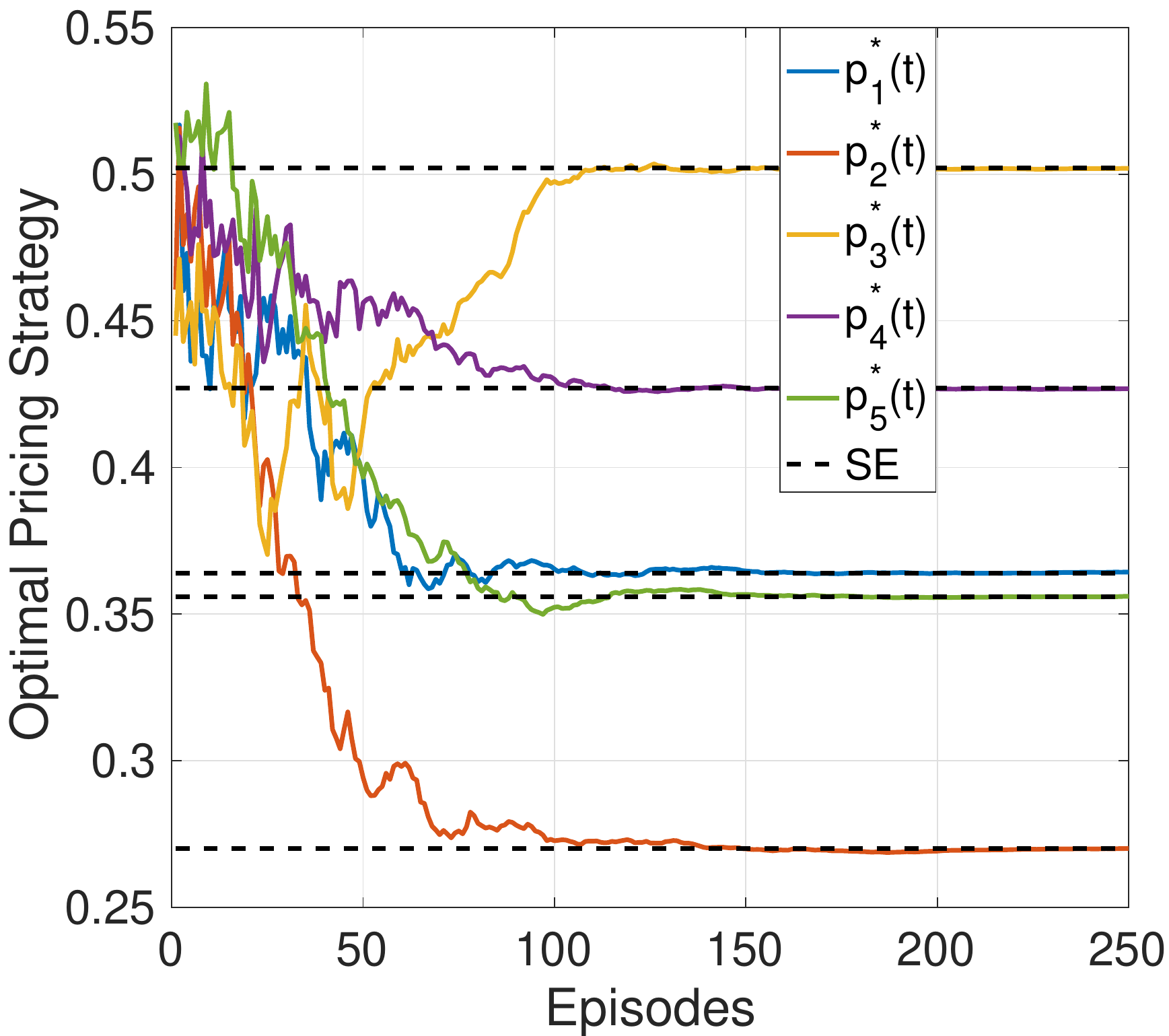}%
\label{fig_6-2-1-1}}
\hfil
\subfloat[Optimal resource allocation strategy.]{\includegraphics[width = 0.25\linewidth]{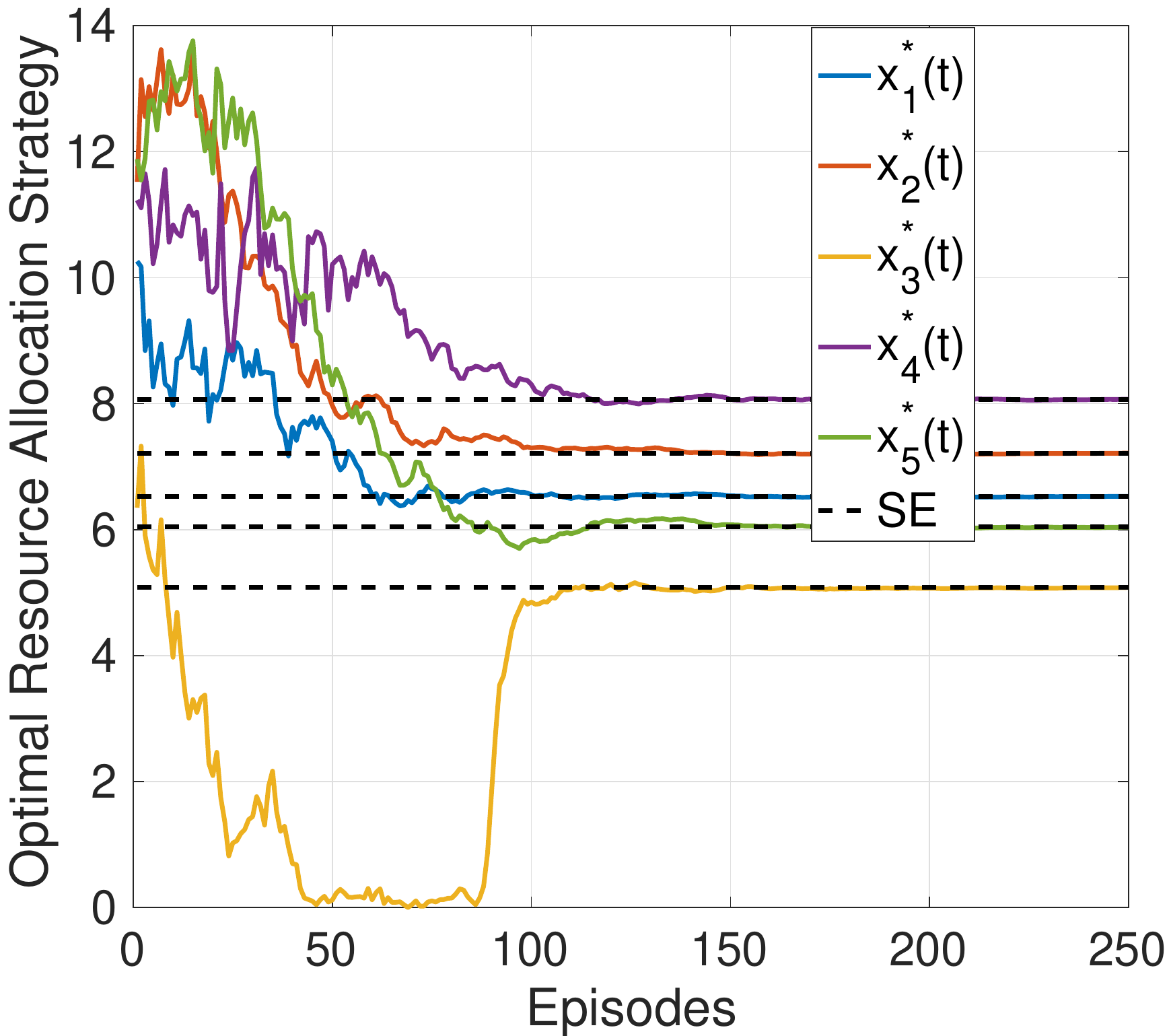}%
\label{fig_6-2-1-2}}
\hfil
\subfloat[SP's payoff.]{\includegraphics[width = 0.25\linewidth]{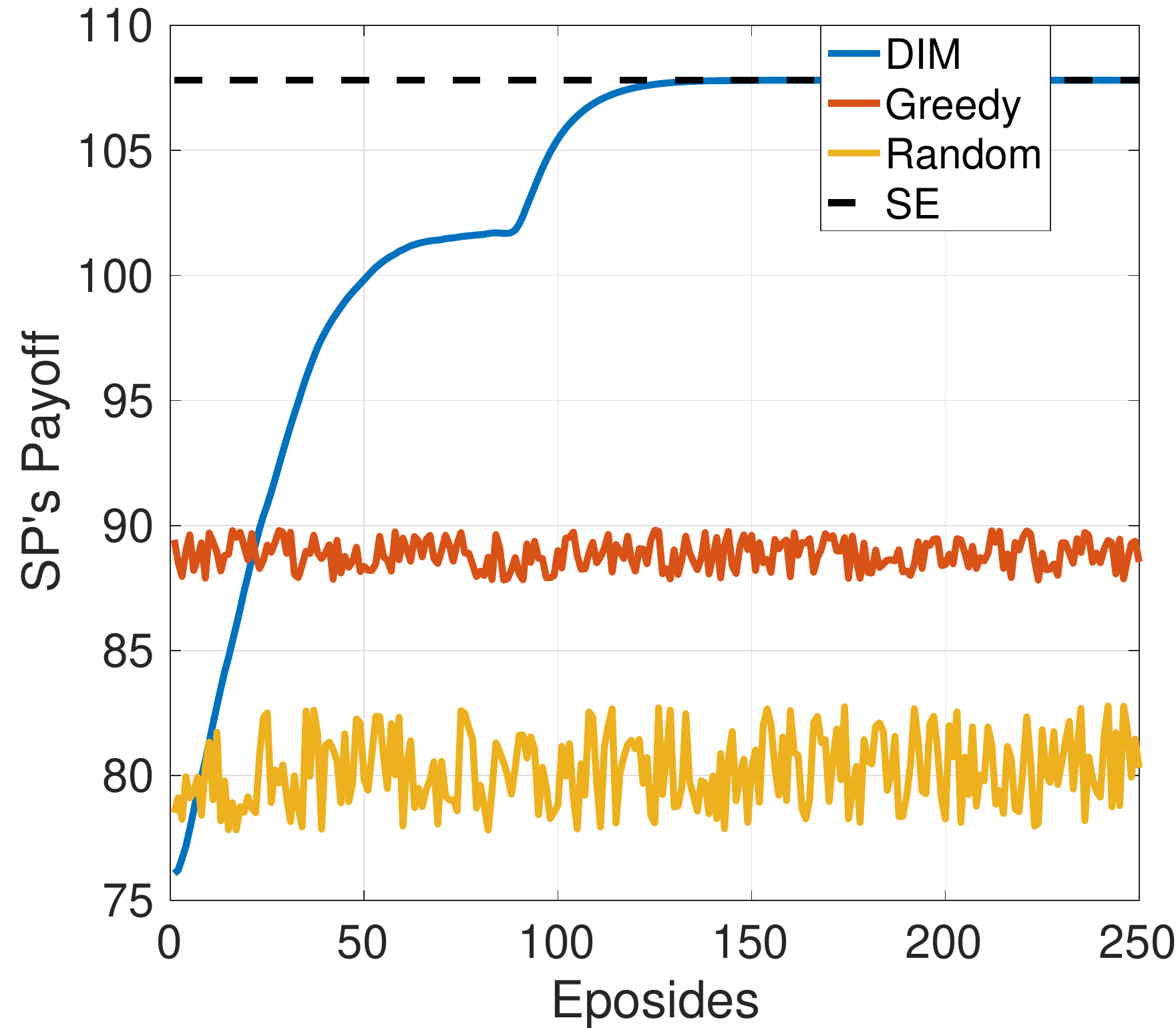}%
\label{fig_6-2-2-1}}
\hfil
\subfloat[MUs' payoff.]{\includegraphics[width = 0.25\linewidth]{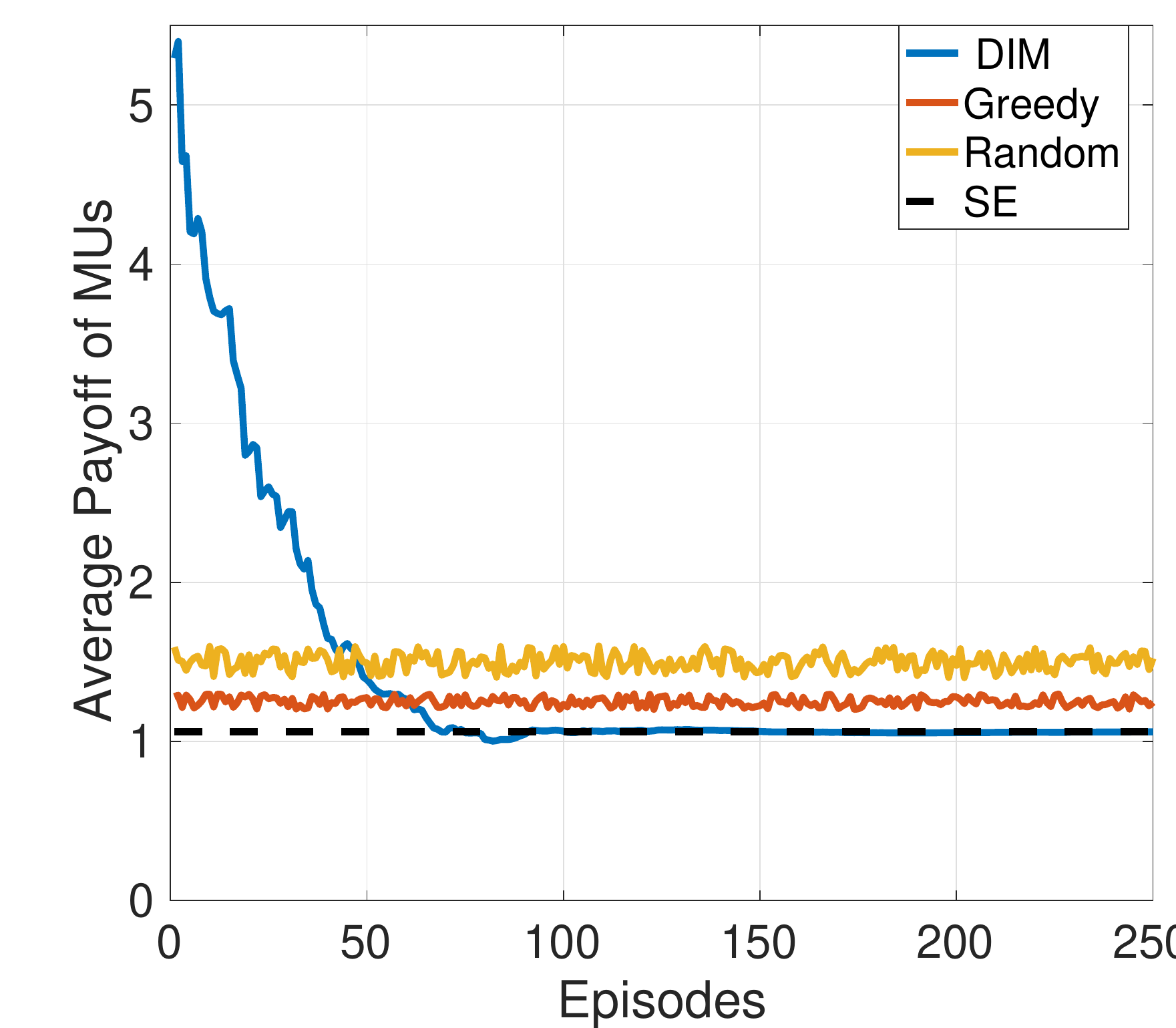}%
\label{fig_6-2-2-2}}
\caption{Performance under dynamic crowdsensing game.}
\label{fig_6-2-1}
\end{figure*}

%\begin{figure*}[!ht]
%\centering
%\subfloat[SP's payoff under dynamic crowdsensing game.]{\includegraphics[width = 0.4\linewidth]{Fig/DRL/SP_payoff.eps}%
%\label{fig_6-2-2-1}}
%\hfil
%\subfloat[MUs' payoff under dynamic crowdsensing game.]{\includegraphics[width = 0.4\linewidth]{Fig/DRL/MUs_payoff.eps}%
%\label{fig_6-2-2-2}}
%\caption{Performance of dynamic crowdsensing game.}
%\label{fig_6-2-2}
%\end{figure*}

\subsection{System Performance}\label{sec_6-2}
Fig.~\ref{fig_6-1-1} shows the system performance under the dynamic crowdsensing game when MUs have different $\delta_n$. In this group of simulation, $c_n = 0$, while $\delta_n$ is randomly chosen from $(0,1]$. From Fig.~\ref{fig_6-1-1}, it can be observed that there is a unique \textbf{SE} between SP and MUs. In Fig.~\ref{fig_6-1-1}, the bar charts denote $\delta_n$. As shown in Fig.~\ref{fig_6-1-1-2}, the $x_n^*$ decreases with $\delta_n$, which implies that the MUs with lower $\delta_n$ will spend more resources to serve the SP. As the MU with lower $\delta_n$, spending resources for serving itself will create little benefit. Therefore, the SP only needs a lower price (as shown in Fig~\ref{fig_6-1-1-1}) can employ more sensing resources from the MUs with lower $\delta_n$. Obviously, this is also in line with the laws of market economy. In Fig.~\ref{fig_6-1-1}, the optimal pricing strategy determined by the SP and the optimal individual resources allocation strategy of each MU have a slow increase with larger value of $\lambda$. It is because the SP with larger gained utilities will pay higher price to MUs and get more sensing resources from MUs, aiming to obtain more payoff.
\begin{figure*}[!ht]
\centering
\subfloat[Optimal pricing strategy profile vs $\delta_n$.]{\includegraphics[width = 0.4\linewidth]{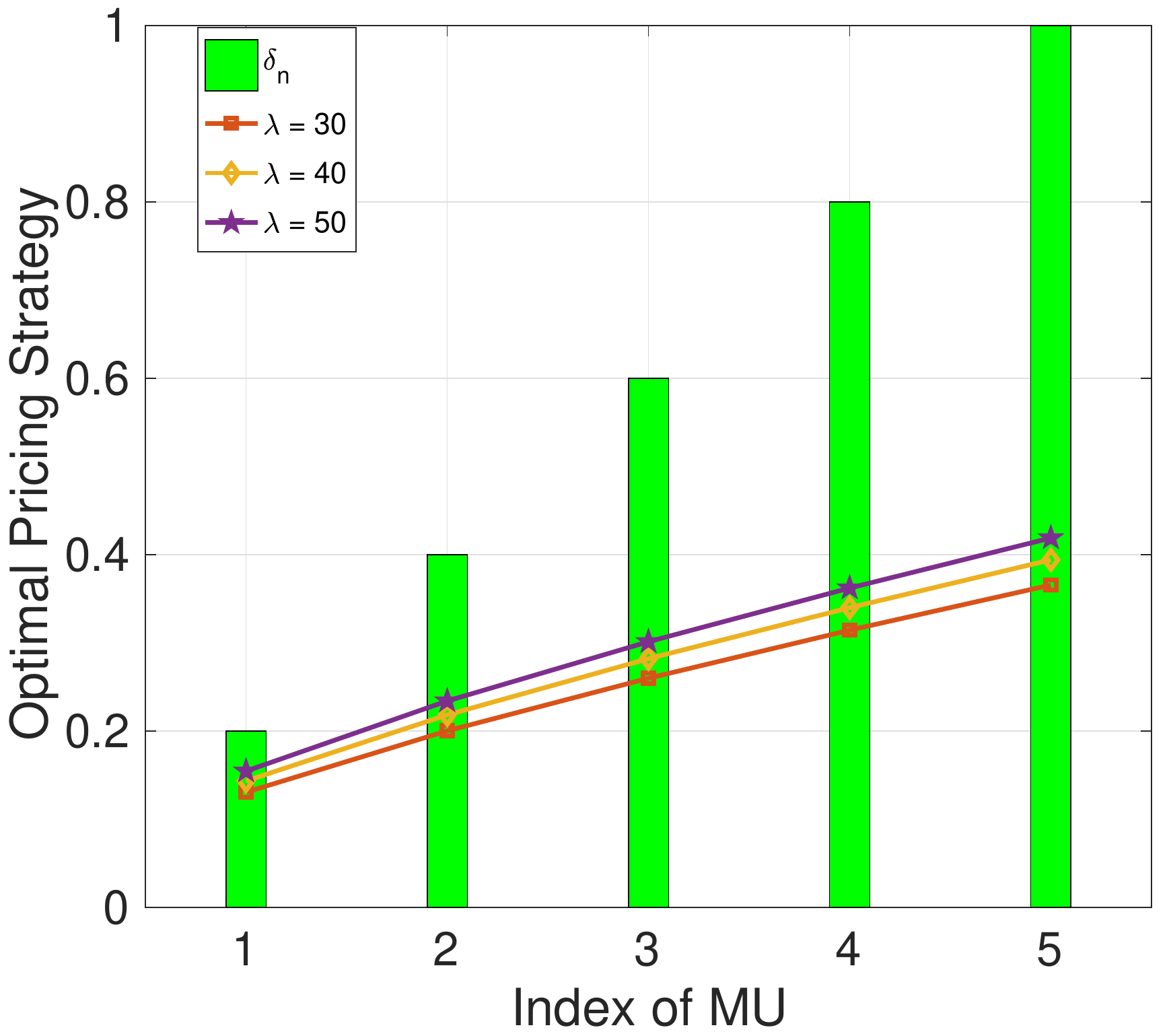}%
\label{fig_6-1-1-1}}
\hfil
\subfloat[Optimal resource allocation strategy profile vs $\delta_n$.]{\includegraphics[width = 0.4\linewidth]{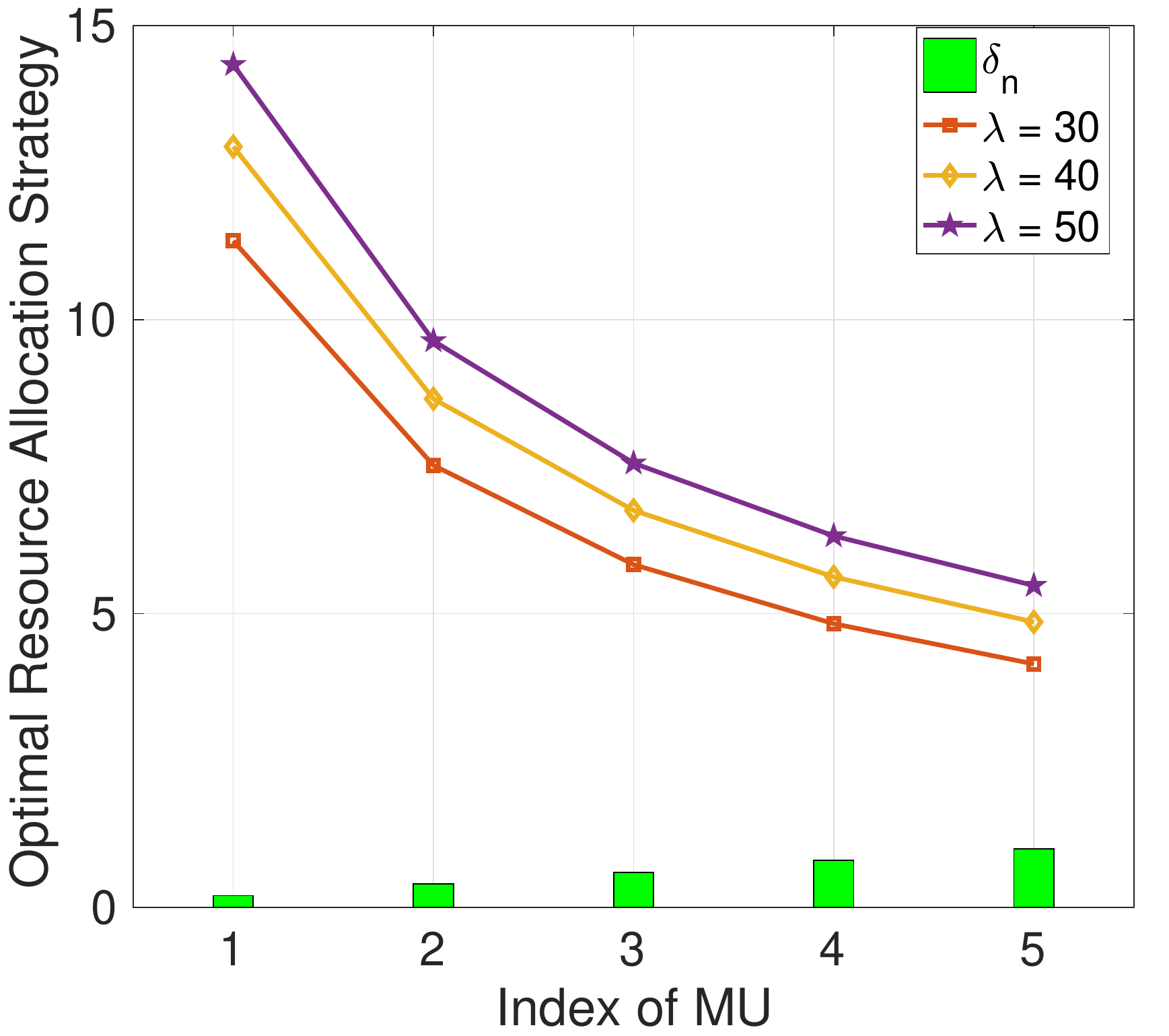}%
\label{fig_6-1-1-2}}
\caption{\textbf{SE} under different $\delta_n$.}
\label{fig_6-1-1}
\end{figure*}

Fig.~\ref{fig_6-1-2} illustrates the \textbf{SE} of MUs with the different $c_n$ under dynamic game. In this group of simulations, $\delta_n=1$, while $c_n$ is randomly chosen from $[0,1)$. We can obtain from this figure that under this setting, there is also a unique \textbf{SE} between SP and MUs. Fig.~\ref{fig_6-1-2-1} shows that under the same $\lambda$, $p_n^*$ will increase with $c_n$, this is because the SP needs to pay the MUs with price higher than the cost $c_n$, or the MUs will not participate in the MCS. Fig.~\ref{fig_6-1-2-2} shows that under the same $\lambda$, the SP will recruit more sensing from the MUs with smaller $c_n$, this is because the SP recruits more sensing resources form MU with smaller $c_n$ will take a little overhead. Also as has mentioned above, when the $\lambda$ is increasing, the MUs with same $c_n$ and $\delta_n$ will allocate more sensing resources to the SP. Meanwhile, the SP will increase the price $p_n^*$.

\begin{figure*}[!ht]
\centering
\subfloat[Optimal pricing strategy profile vs $c_n$.]{\includegraphics[width = 0.4\linewidth]{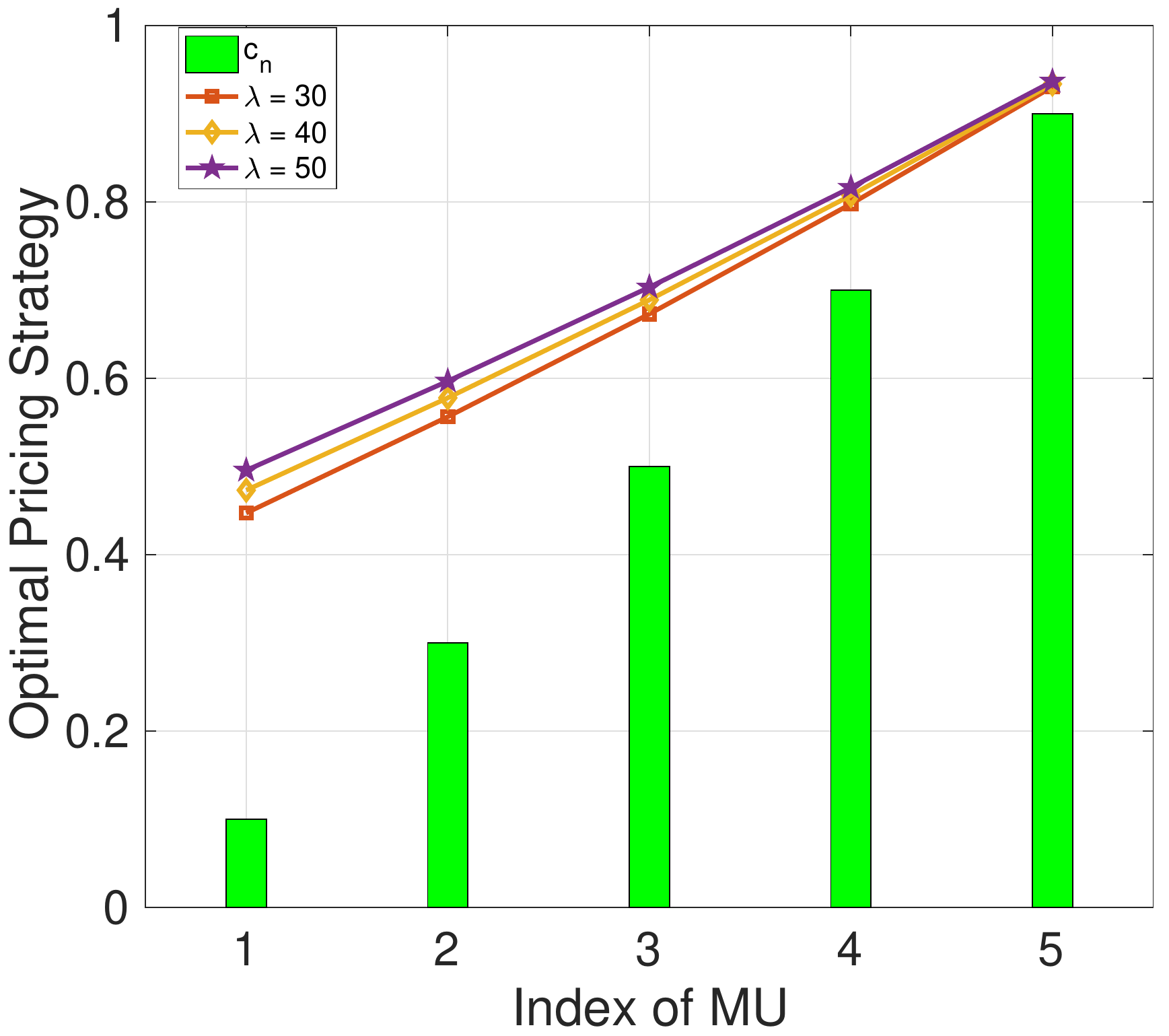}%
\label{fig_6-1-2-1}}
\hfil
\subfloat[Optimal resource allocation strategy profile vs $c_n$.]{\includegraphics[width = 0.4\linewidth]{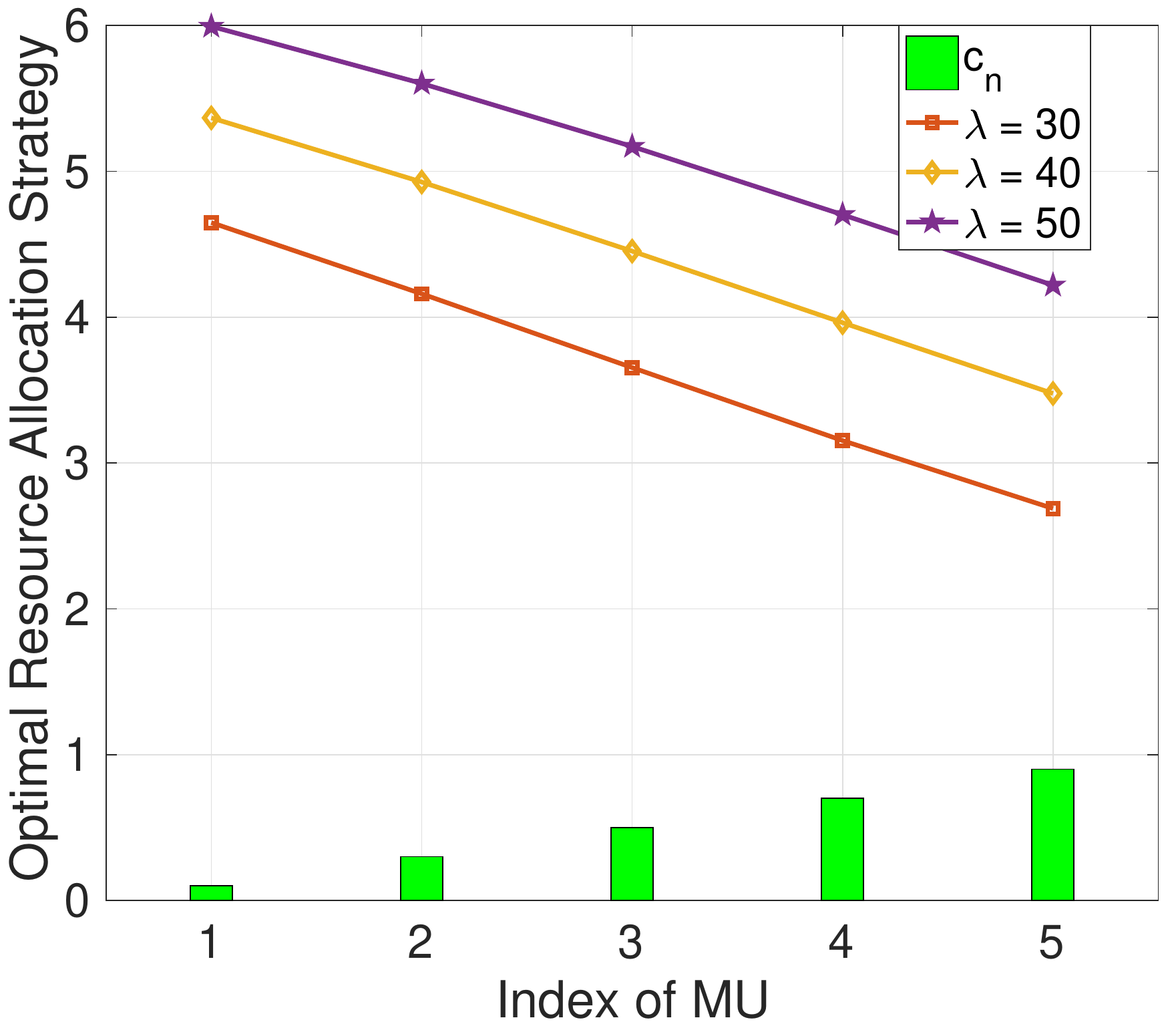}%
\label{fig_6-1-2-2}}
\caption{\textbf{SE} under different $c_n$.}
\label{fig_6-1-2}
\end{figure*}

Fig.~\ref{fig_6-1-3} shows the system performance under the impact of MUs' upper demand. In this group of simulation, we set $\delta_n$ and $c_n$ randomly select from $[0,1]$ and $\delta_n > c_n$, $\lambda=30$ and $\overline {\xi_n}$ varies from $20$ to $30$. The simulation results show that the demand uncertainties have an significant impact on the system performance. When $\overline {\xi_n}$ is higher, its means that the MUs are more expect to use their smart devices to service themselves. In this case, if the SP wants to recruit the MUs to participate in the MCS, it needs to pay more. Fig.~\ref{fig_6-1-3-1} shows that when $\overline {\xi_n}$ is increasing, SP needs to increase the sensing price. Fig.~\ref{fig_6-1-3-2} shows that under this condition, MUs will distribute less resources to participate in the MCS, and left more resources to serve themselves. Obviously, in this case, the payoff of the SP will decrease with $\overline {\xi_n}$ increasing.
\begin{figure*}[!ht]
\centering
\subfloat[Optimal pricing strategy profile vs $\overline {\xi_n}$.]{\includegraphics[width = 0.3\linewidth]{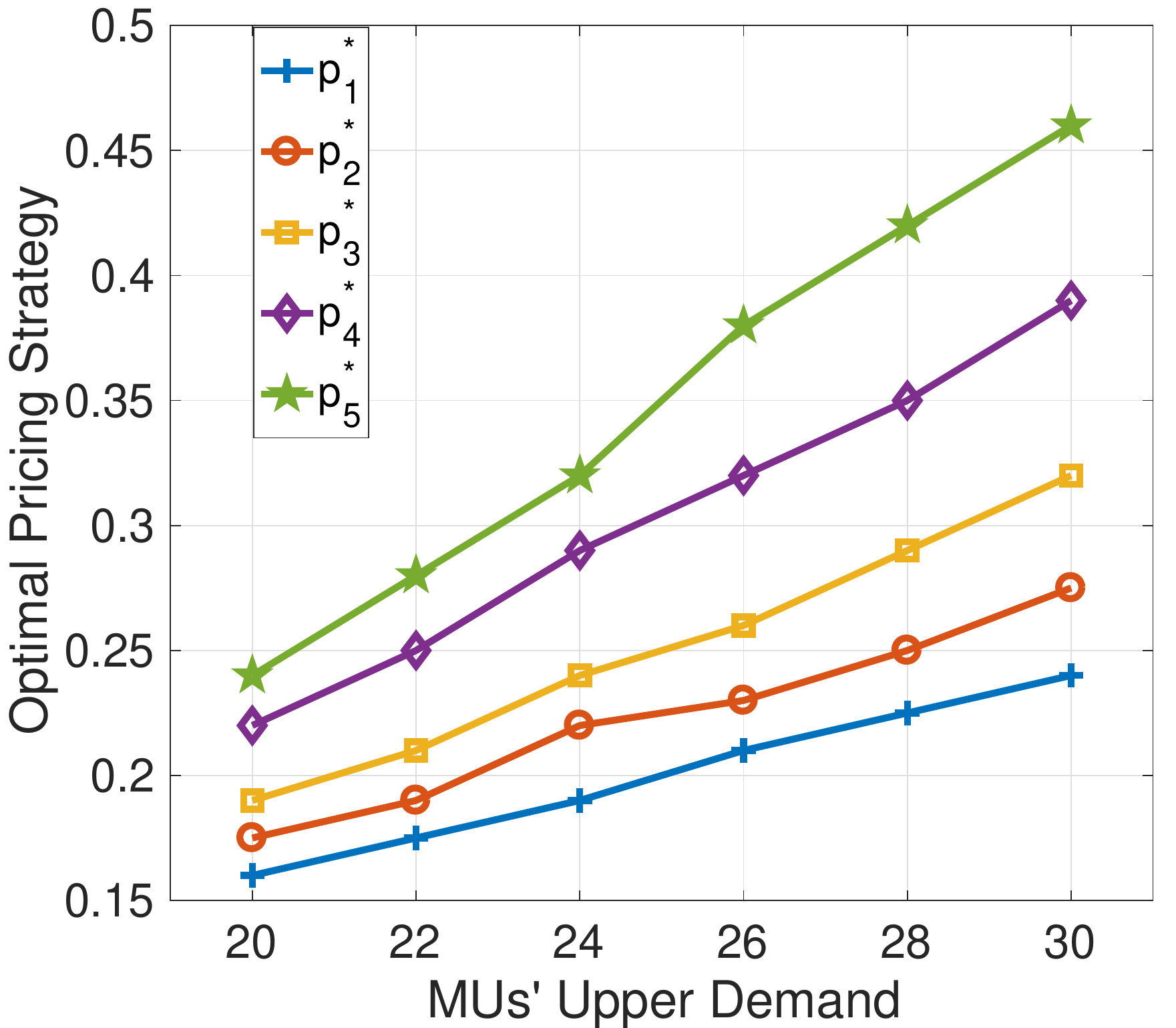}%
\label{fig_6-1-3-1}}
\hfil
\subfloat[Optimal resource allocation strategy profile vs $\overline {\xi_n}$.]{\includegraphics[width = 0.3\linewidth]{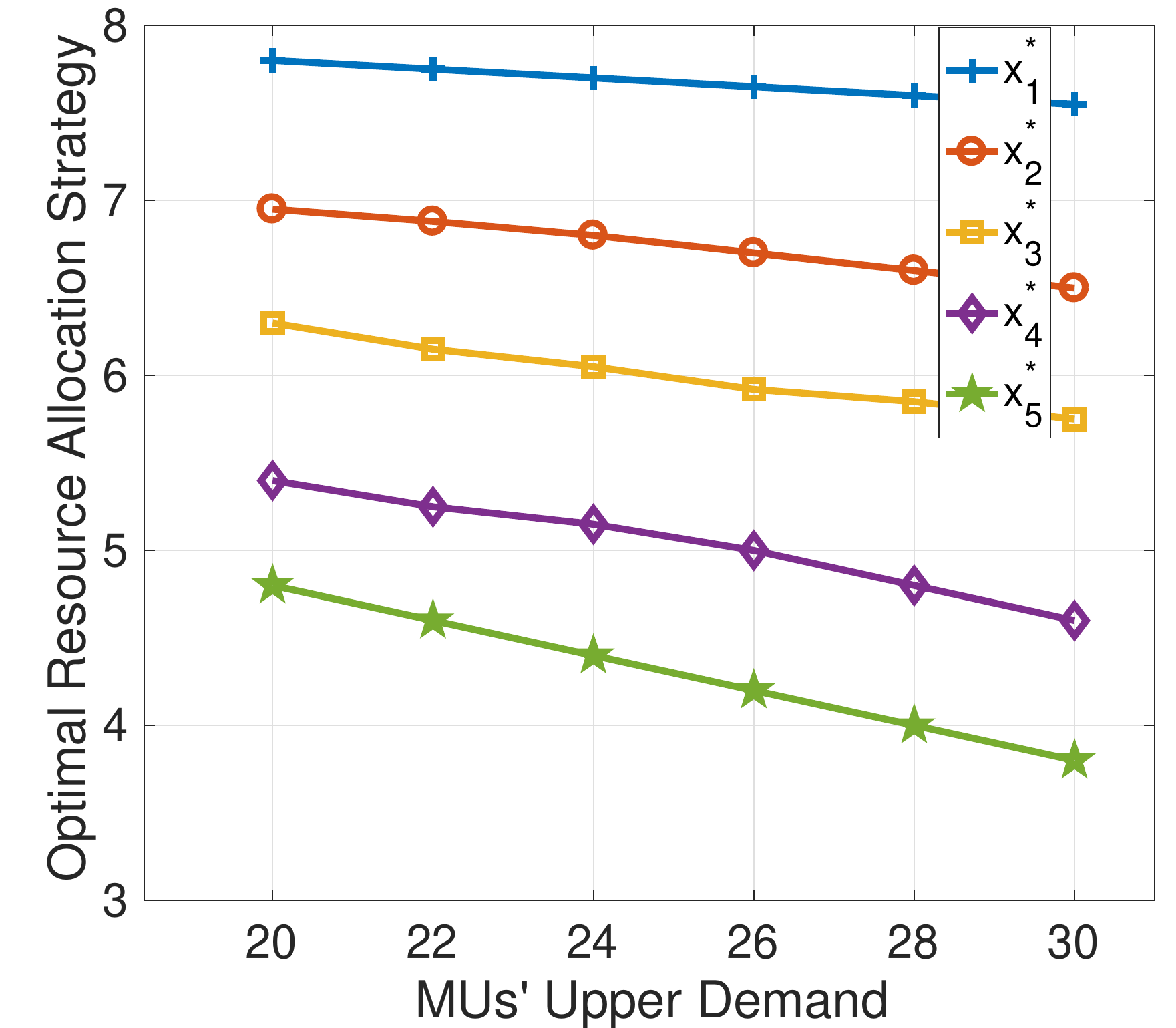}%
\label{fig_6-1-3-2}}
\hfil
\subfloat[SP's payoff vs $\overline {\xi_n}$.]{\includegraphics[width = 0.3\linewidth]{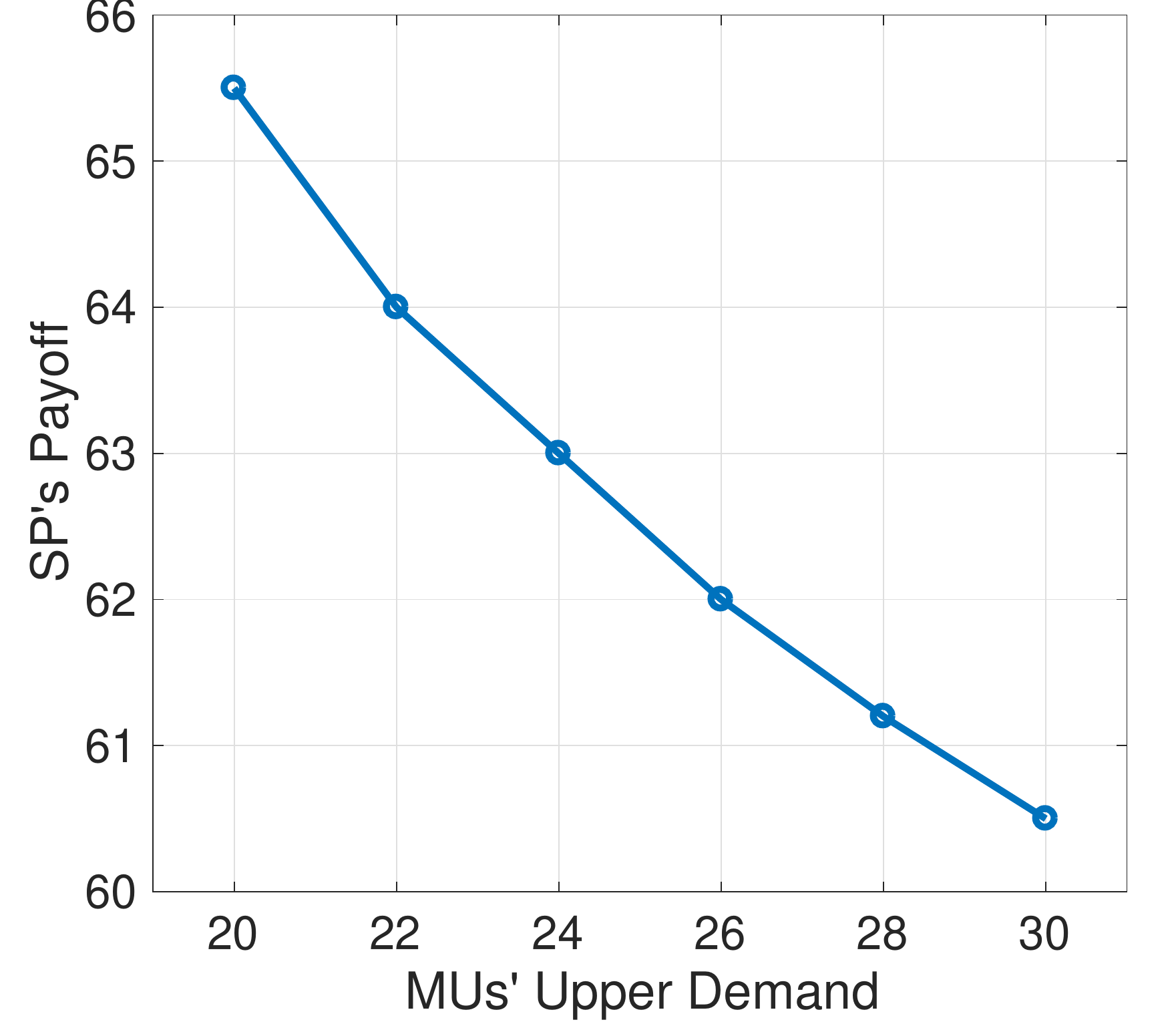}%
\label{fig_6-1-3-3}}
\caption{MCS system performance under the impact of MUs' upper demand.}
\label{fig_6-1-3}
\end{figure*}
\section{Conclusion}\label{sec_7}
In this paper, the static MCS game with MUs' resources constraints and demand uncertainties is formulated firstly, the incentive mechanism is then considered based on a Stackelberg game. The existence of the unique \textbf{SE} is proved and the expressions for calculating the \textbf{SE} are provided. By analyzing the \textbf{SE}, it is found that the MUs' demand uncertainties have evident impacts on the performance of the MCS system. Moreover, considering that the SP requires the MUs' private information to achieve the \textbf{SE} in the static mobile crowdsensing game, a dynamic DRL-based MCS system is proposed, where the SP can obtain the optimal pricing strategy without any prior knowledge of the MUs' information. Therefore, not only MUs can be promoted to participate in the dynamic mobile crowdsensing game, but also the private information of MUs can be kept. Finally, simulation results illustrate that the effectiveness of the proposed mechanism and approach. In future, crowdsensing game with different sensing quality and faked sensing attacks will be researched and the robustness of the DRL approach for private information protection will be enhanced.
\section*{Acknowledgment}
This work was supported by the Beijing Natural Science Foundation under Grant 4161001, the National Natural Science Foundation Projects of International Cooperation and Exchanges under Grant 61720106010, the Foundation for Innovative Research Groups of the National Natural Science Foundation of China under Grant 61621063, the National Natural Science Foundation of China 61572347, the US National Science Foundation (CNS-1319915 and CNS-134335), and the U.S. Department of Transportation Center for Advanced Multimodal Mobility Solutions and Education.

\end{document}